\documentclass[12pt, draftclsnofoot, onecolumn]{IEEEtran}
\ifCLASSINFOpdf

\else

\fi

\usepackage{amsmath}
\usepackage{makeidx}  
\usepackage{algorithm}
\usepackage{algorithmic}
\usepackage{graphicx}
\usepackage{subfigure}
\usepackage{epstopdf}
\usepackage{bm}
\usepackage{bbding}
\usepackage{cite}
\usepackage{stfloats}

\usepackage{amssymb}
\usepackage{graphicx}

\usepackage{url}

\usepackage{tabularx,booktabs}
\newcolumntype{C}{>{\centering\arraybackslash}X} 
\usepackage{lipsum}

\usepackage{makecell} 

\usepackage{graphicx}
\usepackage{color}
\newtheorem{thm}{Theorem}
\newtheorem{rem}{Remark}
\newtheorem{lem}{Lemma}
\newtheorem{pos}{Proposition}
\newtheorem{proof}{proof}

\begin{document}

\title{Fundamental Limits of Intelligent Reflecting Surface Aided Multiuser Broadcast Channel}
\author{Guangji Chen
        and Qingqing~Wu,~\IEEEmembership{Senior Member,~IEEE}
        \thanks{G. Chen is with the Institute for Signal Processing and Systems at Shanghai Jiao Tong University and also with the State Key Laboratory of Internet of Things for Smart City, University of Macau, Macao 999078, China (email: guangjichen@um.edu.mo). Q. Wu is with the Department of Electronic Engineering, Shanghai Jiao Tong University, 200240, China (e-mail: qingqingwu@sjtu.edu.cn). }}

\maketitle

\begin{abstract}
Intelligent reflecting surface (IRS) has recently received significant attention in wireless networks owing to its ability to smartly control the wireless propagation through passive reflection. Although prior works have employed the IRS to enhance the system performance under various setups, the fundamental capacity limits of an IRS aided multi-antenna multi-user system have not yet been characterized. Motivated by this, we investigate an IRS aided multiple-input single-output (MISO) broadcast channel by considering the capacity-achieving dirty paper coding (DPC) scheme and dynamic beamforming configurations. We first propose a bisection based framework to characterize its capacity region by optimally solving the sum-rate maximization problem under a set of rate constraints, which is also applicable to characterize the achievable rate region with the zero-forcing (ZF) scheme. Interestingly, it is rigorously proved that dynamic beamforming is able to enlarge the achievable rate region of ZF if the IRS phase-shifts cannot achieve fully orthogonal channels, whereas the attained gains become marginal due to the reduction of the channel correlations induced by smartly adjusting the IRS phase-shifts. The result implies that employing the IRS is able to reduce the demand for implementing dynamic beamforming. Finally, we analytically prove that the sum-rate achieved by the IRS aided ZF is capable of approaching that of the IRS aided DPC with a sufficiently large IRS in practice. Simulation results shed light on the impact of the IRS on transceiver designs and validate our theoretical findings, which provide useful guidelines to practical systems by indicating that replacing sophisticated schemes with easy-implementation schemes would only result in slight performance loss.

\end{abstract}


\vspace{30pt}

\section{Introduction}
Intelligent reflecting surface (IRS) has recently been proposed as a revolutionary candidate technology for sixth-generation (6G) wireless networks due to its potential to reconfigure the propagation of wireless signals via smart signal reflections \cite{8910627,di2020smart,wu2021intelligent,yuan2022recent,tang2022recent}. Specifically, an IRS is a digitally controllable planar meta-surface consisting of a large number of passive, sub-wavelength sized, and tunable reflecting elements. With an IRS smart controller, each of the reflecting elements is able to induce an independent phase-shift to the incident signal in real time, thereby enabling a flexibly dynamic control over wireless channels \cite{di2020smart}. As such, IRSs are capable of realizing the so-called ``smart radio environment'' via judiciously designing their reflection coefficients, where the power of the intended signal can be improved and that of the interference can be reduced \cite{wu2021intelligent}. Moreover, IRSs enjoy additional appealing advantages such as lightweight, conformal geometry, and low profile. To successfully integrate the IRS in future wireless networks and fully reap the benefits of the IRS, substantial research efforts have been dedicated to address various technical challenges including but not limited to channel estimation, joint optimization of the IRS phase-shifts and wireless resource allocations, and IRS deployment problems \cite{zheng2022survey}.

Among these, joint optimization of the IRS phase-shifts and wireless resource allocations is one crucial problem to maximally reap the performance gains in IRS aided wireless communications. To theoretically capture the receive signal power enhanced by the IRS in a single user case, the seminal work \cite{wu2019beamforming} proved that the fundamental squared power gain of the IRS can be achieved even with only discrete phase-shifts. Owing to this promising result, there have been substantial existing works investigating the problem regarding the joint optimization of the IRS phase-shifts and wireless resource allocations under various system setups, such as multi-cell multiple-input multiple-output (MIMO) networks \cite{9459505, 9279253, zhang2021beyond, pan2020multicell}, non-orthogonal multiple access (NOMA) \cite{mu2020exploiting, chen2020performance, zheng2020intelligent, chen2022active, fu2021reconfigurable}, orthogonal frequency division multiplexing (OFDM) based wireless systems \cite{yang2020intelligent, li2021intelligent}, wireless information and power transfer \cite{pan2020intelligent,wu2021irs,9716123}, mobile edge computing \cite{zhou2020delay, chen2022irs1, chu2020intelligent,bai2020latency,chen2021irs1}, integrated sensing and communication \cite{hua2022joint, meng2022intelligent}, among others. For instance, a dynamic IRS beamforming scheme to maximize the system's minimum rate in an IRS aided OFDM network was studied in \cite{yang2020intelligent}. To fully unleash the potential of the IRS in wireless powered communications networks (WPCNs), the authors in \cite{wu2021irs} further proposed a novel dynamic IRS beamforming framework to strike a balance between the system throughput and the resulting signalling overhead. The results in \cite{wu2021irs} demonstrated that employing dynamic IRS beamforming in WPCNs is able to achieve both high spectral and energy efficiency.


Despite ongoing research progress, prior literature usually employed suboptimal and low-complexity transmission schemes (e.g., time division multiple access (TDMA) or linear-based precoding for multiuser communications \cite{yang2020intelligent, li2021intelligent, wu2021irs, zhou2020delay, chen2022irs1, chu2020intelligent}), and designed generally suboptimal phase-shifts optimization algorithms (e.g., successive convex approximation and semidefinite relaxation \cite{mu2020exploiting, zheng2020intelligent, chen2022active, fu2021reconfigurable, pan2020intelligent}). There still lacks the investigation on the fundamental capacity limits of IRS aided wireless networks with capacity-achieving transmission schemes from an information-theoretic view. This is essential for understanding the theoretical performance upper-bound and guiding system designs. To the authors' best knowledge, there are only handful of works \cite{zhang2020capacity, 9427474, mu2021capacity} that investigated the capacity limits of IRS aided wireless systems. In particular, the authors in \cite{zhang2020capacity} studied the capacity performance of IRS-aided point-to-point MIMO communication systems. It is worth noting that the design of the IRS phase-shifts in a multi-user network needs to incorporate efficient multiple access schemes and balance the performance tradeoff among different users, which is thus fundamentally different from the single-user case. For the two-user setup with single antenna at the BS, the authors in \cite{9427474} derived the outer bound of the capacity regions for both the uplink multiple access channel and downlink broadcast channel under the two-user setup. The results were further extended to the downlink multi-user setup with single-antenna (also termed as ``degraded broadcast channel'' \cite{fundamental}) in \cite{mu2021capacity}. It was demonstrated that the capacity-achieving transmission scheme is to adopt NOMA with dynamic IRS beamforming, i.e., adopting multiple IRS beamforming patterns in a channel coherence interval. However, the works \cite{9427474, mu2021capacity} cannot be applied to the case with multiple antennas at the BS since its corresponding capacity achieving scheme is fundamentally different from that of a single-antenna communication system. Therefore, the capacity limit of the general IRS aided multiple-input single-output (MISO) broadcast channel still remains largely unknown, which thus motivates this work.

In this paper, we investigate the capacity region of an IRS aided MISO broadcast channel from an information-theoretic perspective, where a multi-antenna BS sends independent information to multiple users with the aid of an IRS. To fully harness the benefits brought by the IRS, we consider the dynamic beamforming configuration, where the active beamforming at the BS and the IRS phase-shifts can be adjusted multiple times in a channel coherence interval. For the conventional MISO broadcast channel (also termed as ``non-degraded broadcast channel'' \cite{fundamental}), it has been shown that dirty paper coding (DPC) is capable of achieving the capacity region \cite{weingarten2006capacity}. Due to its prohibitively high complexity, DPC is difficult to be implemented in practical wireless systems. Compared to DPC, zero-forcing (ZF) precoding is easier to be implemented by transmitting information signals in the null space of other users' channels \cite{wang2010linear}, yet ZF generally suffers substantial performance-loss compared to DPC due to the random/uncontrolable wireless propagation. However, for the system with an IRS, properly designing the IRS phase-shifts provides more degrees of freedom (DoFs) for the active beamforming designs at the BS, which can be efficiently exploited to suppress the inter-user interference in the spatial domain and further enhance the desired signal. This may result in different impacts on the performance of IRS aided DPC and ZF.

By considering the effect of favorable time-varying channels proactively induced by the IRS, some fundamental issues still remain unaddressed in the considered IRS aided non-degraded broadcast channel, which are identified as follows.
\begin{itemize}
  \item First, whether employing dynamic beamforming is able to enlarge the capacity and achievable rate regions of the IRS aided DPC and ZF schemes or not?
  \item Second, how much performance loss is caused by adopting the practical linear-precoding based ZF scheme as compared to the non-linear capacity achieving scheme, i.e., DPC?
\end{itemize}
For the first question, it has been demonstrated that exploiting dynamic IRS beamforming in the single-antenna wireless system with the TDMA scheme can significantly improve the sum throughput \cite{wu2021irs}. Regarding the IRS aided non-degraded broadcast channel, the potential gain attained by dynamic beamforming is still unclear since all users are simultaneously served and multiplexed in the spatial domain. Therefore, it remains an open problem whether dynamic beamforming is deserved to be employed or not at the expense of more signalling overhead. The second question is motivated by the fact that ZF can achieve the optimal sum capacity when users' channels are asymptotically orthogonal. Such favorable orthogonal channel conditions can be satisfied in a conventional wireless system without IRS by equipping a large number of antennas at the BS \cite{7402270} or adopting the sophisticated user selection algorithm \cite{yoo2006optimality} when the number of users goes to infinity. However, for the system with an IRS, benefited by the favourable time-varing channels generated by properly designing the IRS phase-shifts, the corresponding orthogonal or semi-orthogonal channels may be naturally obtained for only small number of antennas at the BS without performing user selection, which has non-trivial effects on simplifying the transceiver design. Aiming to tackle the aforementioned issues, the main contributions of this paper are summarized as follows:

\begin{itemize}
  \item First, we consider the capacity-achieving DPC transmission scheme, based on which we characterize the capacity region under the static beamforming by maximizing the sum-rate of all users subject to a set of rate constraints. Although such a sum-rate maximization problem is highly non-convex and challenging to solve, we develop a bisection-based framework to obtain its optimal value by exploiting its special structure. Then, we show that the capacity region under the dynamic beamforming configuration can be obtained by taking the convex hull operation of that under the static beamforming configuration, which serves as a foundation for quantifying the practical performance gap between the IRS aided ZF and DPC.
  \item Moreover, we prove that dynamic beamforming is able to enlarge the achievable rate region of ZF if designing IRS phase-shifts cannot achieve fully orthogonal channels. However, the attained gains would become marginal as the reduction of channel correlations induced by smartly adjusting the IRS phase-shifts. The resulting marginal gain suggests that dynamic beamforming is generally not needed for avoiding additional signalling overhead/implementational complexity incurred by performing dynamic beamforming.
  \item Next, we investigate the performance gap between the IRS aided ZF and DPC. It is rigorously proved that  the relative sum-rate gain achieved by the IRS aided DPC over the IRS aided ZF tends to be zero when the number of IRS elements becomes sufficiently large. This indicates that employing the sub-optimal/low-complexity transmission scheme in practice is capable of approaching the theoretical performance limit, which provides vital guidelines for simplifying transceiver designs.
  \item Finally, numerical results are provided to validate our theoretical findings and to draw useful insights. It is observed that 1) the capacity regions of the IRS aided DPC scheme under the extensively adopted channel setups are convex, which implies that dynamic beamforming is generally not needed; 2) the channel correlation coefficient maintains at a relatively low value after optimizing the IRS phase-shifts for ZF, which leads to marginal gains attained by dynamic beamforming for the IRS aided ZF scheme; and 3) the sum-rate performance of IRS aided ZF approaches that of DPC as the number of the IRS elements increases.
\end{itemize}

The reminder of this paper is organized as follows. Section II introduces the system model of the IRS aided MISO broadcast channel under the dynamic beamforming configuration. Section III presents the capacity and achievable rate regions characterization methods for both the IRS aided DPC and ZF schemes. Section IV provides theoretical analysis regarding dynamic beamforming and the practical performance gap between the IRS aided ZF and DPC. Section V presents numerical results to verify our theoretical findings and draw useful insights. Finally, Section VI concludes the paper.

\emph{Notations:} Boldface upper-case and lower-case  letter denote matrix and   vector, respectively.  ${\mathbb C}^ {d_1\times d_2}$ stands for the set of  complex $d_1\times d_2$  matrices. For a complex-valued vector $\bf x$, ${\left\| {\bf x} \right\|}$ represents the  Euclidean norm of $\bf x$, ${\rm arg}({\bf x})$ denotes  the phase of   $\bf x$, and ${\rm diag}(\bf x) $ denotes a diagonal matrix whose main diagonal elements are extracted from vector $\bf x$.
For a vector $\bf x$, ${\bf x}^*$ and  ${\bf x}^H$  stand for  its conjugate and  conjugate transpose respectively.   For a square matrix $\bf X$,  ${\rm{Tr}}\left( {\bf{X}} \right)$, $\left\| {\bf{X}} \right\|_2$ and ${\rm{rank}}\left( {\bf{X}} \right)$ respectively  stand for  its trace, Euclidean norm and rank,  while ${\bf{X}} \succeq {\bf{0}}$ indicates that matrix $\bf X$ is positive semi-definite. A circularly symmetric complex Gaussian random variable $x$ with mean $ \mu$ and variance  $ \sigma^2$ is denoted by ${x} \sim {\cal CN}\left( {{{\mu }},{{\sigma^2 }}} \right)$.  ${\cal O}\left(  \cdot  \right)$ is the big-O computational complexity notation. ${\mathop{\rm Conv}\nolimits} \left( {\cal X} \right)$ denotes the convex hull operation of the set ${\cal X}$. $ \cup $ represents the union operation. ${\mathbb{E}}\left[  \cdot  \right]$ represents the expectation operation.
\vspace{-8pt}
\section{System Model}
\begin{figure}[!t]
\setlength{\abovecaptionskip}{-5pt}
\setlength{\belowcaptionskip}{-5pt}
\centering
\includegraphics[width= 0.5\textwidth]{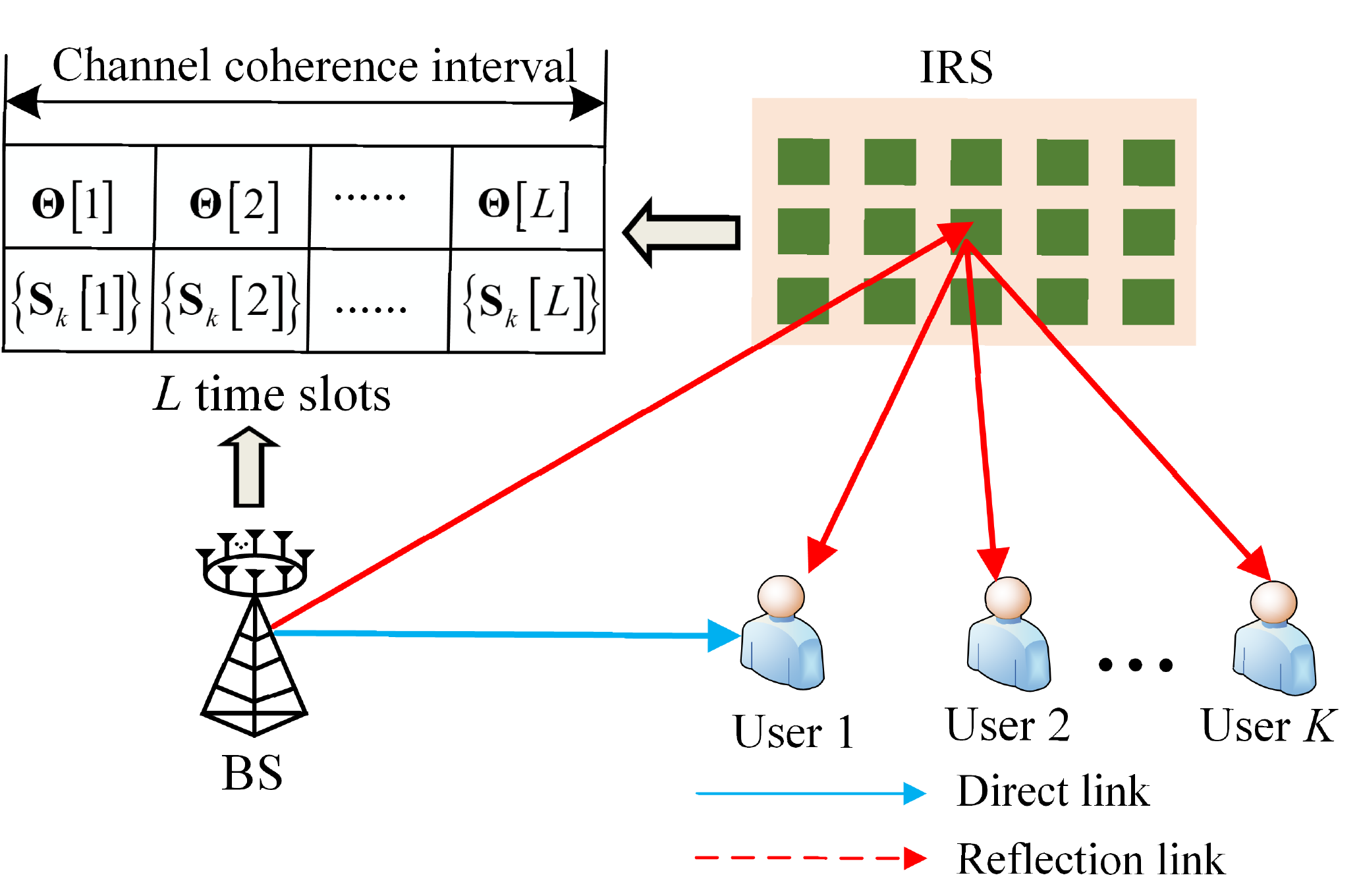}
\DeclareGraphicsExtensions.
\caption{An IRS aided MISO broadcast channel.}
\label{model}
\vspace{-12pt}
\end{figure}
As shown in Fig. \ref{model}, we consider an IRS aided MISO broadcast channel, where a multi-antenna BS sends individual messages to $K$ single-antenna users with the aid of an IRS. Similar to \cite{mu2021capacity}, the IRS with ${N_R}$ reflecting elements is further partitioned into $N$ subsurfaces for ease of practical control. In each subsurface, $\bar N = {N_R}/N$ adjacent elements share a common phase-shift. The set of the users and the number of the antennas at the BS are denoted by ${\cal K} \buildrel \Delta \over = \left\{ {1, \ldots K} \right\}$ and $M$, respectively. To characterize the theoretical capacity gain of the considered system, we assume that perfect CSI of all channels is available at the BS based on the channel acquisition methods proposed in \cite{guan2022irs}.

We focus on a channel coherence interval ${\cal T} \buildrel \Delta \over = \left( {0,T} \right]$ with duration $T > 0$. The channels from the BS to the IRS, from the BS to user $k$, and from the IRS to user $k$ are denoted by ${\bf{G}} \in \mathbb{C}^{N\times M}$, ${\bf{h}}_{d,k}^H \in \mathbb{C}^{1 \times M}$, and ${\bf{h}}_{r,k}^H \in \mathbb{C}^{1 \times N}$, respectively. In addition, we assume that all channels follow the quasi-static flat-fading model, which indicates that ${\bf{G}}$, ${\bf{h}}_{d,k}^H$, and ${\bf{h}}_{r,k}^H $ remain unchanged for $\forall t \in {\cal T}$. To fully harness the potential gain brought by the IRS, the IRS beamforming pattern can be reconfigured multiple times in a channel coherence interval. Specifically, we divide the whole channel coherence interval ${\cal T}$ into $L$ time slots (TSs) and the $l$-th TS is denoted by ${{\cal T}_l} \buildrel \Delta \over = \left( {\sum\nolimits_{i = 1}^{l - 1} {{t_i},\sum\nolimits_{i = 1}^l {{t_i}} } } \right]$, $\forall l \in {\cal L} \buildrel \Delta \over = \left\{ {1, \ldots L} \right\}$. The duration of the $l$-th TS is ${t_l}$ and $\sum\nolimits_{l = 1}^L {{t_l}}  = T$. In the $l$-th TS, its associated IRS beamforming pattern, denoted by ${\bf{\Theta }}\left[ l \right] = {\mathop{\rm diag}\nolimits} \left( {{e^{j{\theta _1}\left[ l \right]}}, \ldots ,{e^{j{\theta _N}\left[ l \right]}}} \right)$, is employed to assist downlink transmission, where ${\theta _n}\left[ l \right] \in \left[ {0,2\pi } \right)$ is the phase-shift of the $n$-th sub-surface. For ease of practical implementation, we consider the finite resolution for each IRS element and thus the discrete unit-modulus constraint ${\theta _n}\left[ l \right] \in {\cal F} \buildrel \Delta \over = \left\{ {{{2\pi q} \mathord{\left/
 {\vphantom {{2\pi q} {Q,q = 0,1, \ldots ,Q - 1}}} \right.
 \kern-\nulldelimiterspace} {Q,q = 0,1, \ldots ,Q - 1}}} \right\}$ is imposed, where $Q = {2^b}$ and $b$ denotes the number of bits used to quantize the number of phase-shift levels.

Without loss of generality, in the $l$-th TS, the BS transmits $K$ independent information signals, i.e., ${{\bf{x}}_k}\left[ l \right]\in \mathbb{C}^{M\times 1}$, $\forall k \in {\cal K},l \in {\cal L}$, for each user. It is assumed that a Gaussian codebook is used for each information signal and ${{\bf{x}}_k}\left[ l \right]\sim {\cal C}{\cal N}\left( {0,{{\bf{S}}_k}\left[ l \right]} \right)$, $\forall k \in {\cal K}$, where ${{\bf{S}}_k}\left[ l \right] = {\mathbb{E}}\left[ {{{\bf{x}}_k}\left[ l \right]{\bf{x}}_k^H\left[ l \right]} \right]$ is the covariance matrix of the signal for user $k$. Thus, the $M$-dimensional complex baseband signal transmitted by the BS in the $l$-th TS is ${\bf{x}}\left[ l \right] = \sum\nolimits_{k = 1}^K {{{\bf{x}}_k}\left[ l \right]} ,\forall l \in {\cal L}$. Suppose that the maximum transmit power at the BS is ${P_{\max }}$. Then, we have
\begin{align}\label{power constraint}
{\mathbb{E}}\left[ {{\bf{x}}\left[ l \right]{{\bf{x}}^H}\left[ l \right]} \right]  = \sum\nolimits_{k = 1}^K {{\rm{Tr}}\left( {{{\bf{S}}_k}\left[ l \right]} \right)}  \le {P_{\max }},\forall l \in {\cal L}.
\end{align}
\vspace{-8pt}
\subsection{IRS aided DPC Scheme}
To characterize the fundamental capacity limit of the IRS-aided MISO broadcast channel, the information-theoretically optimal DPC is considered for the information transmission, where the causal interference is pre-cancelled at the BS \cite{weingarten2006capacity}. Specifically, the information signal ${{\bf{x}}_k}\left[ l \right]$ for user $k$ is encoded before each user $i$ with $i > k$. Then, for user $k$, the causal interference caused by each user $i$ with $i < k$ can be canceled at the BS by employing DPC. In the $l$-th TS, the received signal at user $k$ can be written as
\begin{align}\label{received_signal}
{y_k}\left[ l \right] = {\bf{h}}_{k}^H\left( {{\bf{\Theta }}\left[ l \right]} \right){{\bf{x}}_{k}}\left[ l \right] + \sum\nolimits_{i > k} {{\bf{h}}_k^H\left( {{\bf{\Theta }}\left[ l \right]} \right){{\bf{x}}_i}\left[ l \right]}  + {z_k}\left[ l \right], \forall k \in {\cal K},l \in {\cal L}.
\end{align}
where ${\bf{h}}_k^H\left( {{\bf{\Theta }}\left[ l \right]} \right) = {\bf{h}}_{d,k}^H + {\bf{h}}_{r,k}^H{\bf{\Theta }}\left[ l \right]{\bf{G}}$ denotes the equivalent channel from the BS to user $k$ and ${z_k}\left[ l \right]\sim {\cal C}{\cal N}\left( {0,{\sigma ^2}\left[ l \right]} \right)$ denotes the additive white Gaussian noise at user $k$. With Gaussian codebook employed, the achievable rate (in bps/Hz) of user $k$ in the $l$-th TS is given by
\begin{align}\label{rate_l}
{R_k}\left[ l \right] = {\log _2}\left( {\frac{{{\sigma ^2} + {\bf{h}}_k^H\left( {{\bf{\Theta }}\left[ l \right]} \right)\left( {\sum\nolimits_{i = k}^K {{{\bf{S}}_i}\left[ l \right]} } \right){{\bf{h}}_k}\left( {{\bf{\Theta }}\left[ l \right]} \right)}}{{{\sigma ^2} + {\bf{h}}_k^H\left( {{\bf{\Theta }}\left[ l \right]} \right)\left( {\sum\nolimits_{i = k + 1}^K {{{\bf{S}}_i}\left[ l \right]} } \right){{\bf{h}}_k}\left( {{\bf{\Theta }}\left[ l \right]} \right)}}} \right), \forall k \in {\cal K},l \in {\cal L}.
\end{align}
Then, the average achievable rate of user $k$ over the whole channel coherence interval is ${{\bar R}_k} = \frac{1}{T}\sum\nolimits_{l = 1}^L {{t_l}} {R_k}\left[ l \right]$. The feasible region of $\left\{ {{\bf{\Theta }}\left[ l \right],{{\bf{S}}_k}\left[ l \right],{t_l},\forall k \in {\cal K},l \in {\cal L}} \right\}$ is denoted by
\begin{align}\label{feasible1}
{{\cal X}^L} = \left\{ {\begin{array}{*{20}{l}}
{{\theta _n}\left[ l \right] \in {\cal F},{{\bf{S}}_k}\left[ l \right] \succeq{\bf{0}},\sum\nolimits_{l = 1}^L {{t_l} \le T} }\\
{\sum\nolimits_{k = 1}^K {{\rm{Tr}}} \left( {{{\bf{S}}_k}\left[ l \right]} \right) \le {P_{\max }},{t_l} \ge 0}
\end{array}} \right\}.
\end{align}
Accordingly, the capacity region of the IRS-aided MISO broadcast channel achieved by DPC is defined as
\begin{align}\label{region_defination}
{\cal C}\left( L \right) \buildrel \Delta \over = \bigcup\limits_{\left\{ {{\bf{\Theta }}\left[ l \right],{{\bf{S}}_k}\left[ l \right],{t_l}} \right\} \in {{\cal X}^L}} {\tilde {\cal C}\left( {{\bf{\Theta }}\left[ l \right],{{\bf{S}}_k}\left[ l \right],{t_l}} \right)} ,
\end{align}
where $\tilde {\cal C}\left( {{\bf{\Theta }}\left[ l \right],{{\bf{S}}_k}\left[ l \right],{t_l}} \right) = \left\{ {{\bf{\bar r}}:0 \le {r_k} \le {{\bar R}_k}} \right\}$ denotes the achievable average rate tuples of $K$ users under the given $\left\{ {{\bf{\Theta }}\left[ l \right],{{\bf{S}}_k}\left[ l \right],{t_l}} \right\}$ and ${\bf{\bar r}} \buildrel \Delta \over = {\left[ {{r_1}, \ldots ,{r_K}} \right]^T}$.

\vspace{-8pt}
\subsection{IRS aided ZF Scheme}
For the ZF based transmission scheme, one dedicated beamforming vector is assigned for each user in the $l$-th TS. Then, the $M$-dimensional complex baseband
signal transmitted by the BS in the $l$-th TS can be expressed as ${\bf{x}}\left[ l \right] = \sum\nolimits_{k = 1}^K {{{\bf{w}}_k}\left[ l \right]{s_k}} ,\forall l \in {\cal L}$, where ${{{\bf{w}}_k}\left[ l \right]}\in \mathbb{C}^{M\times 1}$ denotes the beamforming vector for user $k$ in the $l$-th TS and ${s_k}$ is a scalar which represents the information-bearing symbol of user $k$. It is assumed that ${s_k}$'s are independent random variables with zero mean and unit variance (normalized power). Based on ZF criterion, the beamforming vector for any user should lie in the null space of the other users' channels and thus ${\bf{h}}_j^H\left( {{\bf{\Theta }}\left[ l \right]} \right){{\bf{w}}_k} = 0,\forall j \ne k$ is satisfied.
Hence, the achievable rate of user $k$ in the $l$-th TS can be expressed as
\begin{align}\label{rate_lb}
R_k^{{\rm{zf}}}\left[ l \right] = {\log _2}\left( {1 + \frac{{{{\left| {{\bf{h}}_k^H\left( {{\bf{\Theta }}\left[ l \right]} \right){{\bf{w}}_k}\left[ l \right]} \right|}^2}}}{{{\sigma ^2}}}} \right),\forall k \in {\cal K},l \in {\cal L}.
\end{align}
Similarly, the average achievable rate of user $k$ under the ZF based transmission scheme over the whole channel coherence interval can be written as
\begin{align}\label{average_rate2}
\bar R_k^{{\rm{zf}}} = \frac{1}{T}\sum\nolimits_{l = 1}^L {{t_l}} R_k^{{\rm{zf}}}\left[ l \right],k \in {\cal K}.
\end{align}
Let ${\cal X}_{{\mathop{\rm zf}\nolimits} }^L$ denote the feasible set of $\left\{ {{\bf{\Theta }}\left[ l \right],{{\bf{w}}_k}\left[ l \right],{t_l},\forall k \in {\cal K},l \in {\cal L}} \right\}$, which is given by
\begin{align}\label{feasible2}
{\cal X}_{{\mathop{\rm zf}\nolimits} }^L \!\!=\!\! \left\{ {{\theta _n}\left[ l \right] \in {\cal F},\sum\limits_{l = 1}^L {{t_l} \le T,{{\sum\nolimits_{k = 1}^K {\left\| {{{\bf{w}}_k}\left[ l \right]} \right\|} }^2} \le {P_{\max }}} ,{\bf{h}}_j^H\left( {{\bf{\Theta }}\left[ l \right]} \right){{\bf{w}}_k} \!=\! 0,\forall j \ne k,{t_l} \ge 0} \right\}.
\end{align}

According to \eqref{average_rate2} and \eqref{feasible2}, we further define the achievable rate region of the IRS-aided MISO broadcast channel under the ZF based transmission scheme as
\begin{align}\label{region_defination1}
{\cal R}\left( L \right) \buildrel \Delta \over = \bigcup\limits_{\left\{ {{\bf{\Theta }}\left[ l \right],{{\bf{w}}_k}\left[ l \right],{t_l}} \right\} \in {\cal X}_{{\mathop{\rm zf}\nolimits} }^L} {\tilde {\cal R}\left( {{\bf{\Theta }}\left[ l \right],{{\bf{w}}_k}\left[ l \right],{t_l}} \right)} ,
\end{align}
where $\tilde {\cal R}\left( {{\bf{\Theta }}\left[ l \right],{{\bf{w}}_k}\left[ l \right],{t_l}} \right) = \left\{ {{\bf{\bar r}}:0 \le {r_k} \le \bar R_k^{{\rm{zf}}}} \right\}$ denotes the achievable average rate tuples of $K$ users under the given $\left\{ {{\bf{\Theta }}\left[ l \right],{{\bf{w}}_k}\left[ l \right],{t_l}} \right\}$ for the ZF based transmission scheme.

\begin{rem}
Under any given channel conditions, it has been shown in \cite{weingarten2006capacity} that the capacity region of the MISO broadcast channel can be achieved by using DPC, whereas it is difficult to be applied in practical wireless systems, due to its prohibitively high complexity. In contrast, ZF-based transmission scheme is much easier to be implemented yet with performance loss.
\end{rem}

After integrating the IRS, whether dynamic beamforming configuration is capable of providing performance gains or not and the practical performance gap between the IRS aided ZF and DPC still remain unknown. In the following, we first provide a framework to characterize the capacity and achievable rate regions of the IRS aided DPC and ZF, respectively. Then, theoretical analysis regarding the performance gains attained by dynamic beamforming and the performance gap between the IRS aided ZF and DPC is provided.
\section{Capacity and Achievable Rate Regions Characterization }
In this section, we aim to characterize the fundamental capacity region of the IRS aided broadcast channel defined in \eqref{region_defination}. However, it still remains unknown how many TSs, i.e., $L$ are needed to achieve the capacity upper bound. Even though the required number of TSs for achieving performance upper-bound, denoted by ${L_u}$, is given, deriving the capacity region is still prohibitively complex since the dimension of all possible combinations of IRS beamforming patterns increases exponentially with ${L_u}N$. To circumvent this issue, we propose an efficient method to characterize ${\cal C}\left( {{L_u}} \right)$ by invoking the rate-profile technique.
\vspace{-10pt}
\subsection{Rate-Profile Based Capacity Characterization}
To start with, we first characterize the capacity region achieved by using a single TS, i.e., ${\cal C}\left( 1 \right)$. Then, we show that the capacity region by using multiple TSs, i.e., ${\cal C}\left( {{L_u}} \right)$, can be reconstructed based on ${\cal C}\left( 1 \right)$. For characterizing ${\cal C}\left( 1 \right)$, we can drop the TS index $l$ since all variables remain static during the whole channel coherence interval. Based on the rate-profile technique, we introduce the following sum-rate maximization problem by jointly optimizing the IRS phase-shifts and the transmit covariance matrix at the BS, i.e.,
\begin{subequations}\label{C1}
\begin{align}
\label{C1-a}\mathop {\max }\limits_{{\bf{\Theta }},\left\{ {{{\bf{S}}_k}} \right\},R}  \;\;&R\\
\label{C1-b}{\rm{s.t.}}\;\;\;\;&{\log _2}\left( {1 + \frac{{{\bf{h}}_k^H\left( {\bf{\Theta }} \right){{\bf{S}}_k}{{\bf{h}}_k}\left( {\bf{\Theta }} \right)}}{{{\bf{h}}_k^H\left( {\bf{\Theta }} \right)\left( {\sum\nolimits_{i > k}^K {{{\bf{S}}_i}} } \right){{\bf{h}}_k}\left( {\bf{\Theta }} \right) + {\sigma ^2}}}} \right) \ge {\alpha _k}R, ~\forall {k} \in {{\cal K}},\\
\label{C1-c}&\sum\nolimits_{k = 1}^K {{\mathop{\rm Tr}\nolimits} } \left( {{{\bf{S}}_k}} \right) \le {P_{\max }},\\
\label{C1-d}&{{\bf{S}}_k} \succeq {\bf{0}}, ~\forall {k} \in {{\cal K}},\\
\label{C1-e}&{\left[ {\bf{\Theta }} \right]_{n,n}} \in {\cal F},~\forall n,
\end{align}
\end{subequations}
where ${\alpha _k} \in \left[ {0,1} \right]$ denotes the rate ratio between user $k$ and the sum-rate, which satisfies $\sum\nolimits_{k = 1}^K {{\alpha _k}}  = 1$. Let ${\bm{\alpha }} = {\left[ {{\alpha _1}, \ldots ,{\alpha _K}} \right]^T}$ and the set of the region for ${\bm{\alpha }}$ is denoted by ${\cal A} \buildrel \Delta \over = \left\{ {0 \le {\alpha _k} \le 1,\forall k,\sum\nolimits_{k = 1}^K {{\alpha _k}}  = 1} \right\}$. Under the given ${\bm{\alpha }}$, the optimal objective value of problem \eqref{C1} is denoted by ${R^*}\left( {\bm{\alpha }} \right)$ and the corresponding rate tuple is ${{\bf{r}}^*}\left( {\bm{\alpha }} \right) = {\left[ {r_1^*\left( {\bm{\alpha }} \right), \ldots ,r_K^*\left( {\bm{\alpha }} \right)} \right]^T} = {R^*}\left( {\bm{\alpha }} \right){\bm{\alpha }}$. Then, we show that the capacity region, i.e., ${\cal C}\left( {{L_u}} \right)$, can be fully characterized based on the rate tuple ${{\bf{r}}^*}\left( {\bm{\alpha }} \right)$ in the following proposition.
\begin{pos}
The capacity region of the IRS-aided MISO broadcast channel under the dynamic beamforming configuration, i.e., ${\cal C}\left( {{L_u}} \right)$, is given by
\begin{align}\label{region_construction}
{\cal C}\left( {{L_u}} \right) = {\mathop{\rm Conv}\nolimits} \left( {{{\bf{0}}_{K \times 1}}\bigcup\nolimits_{{\bm{\alpha }},{\bm{\alpha }} \in {\cal A}} {{R^*}\left( {\bm{\alpha }} \right){\bm{\alpha }}} } \right).
\end{align}
\end{pos}
\begin{proof}
The proof is similar to \cite{9427474}. Interested readers may refer to \cite{9427474} for more details.
\end{proof}

Proposition 1 indicates that the capacity region ${\cal C}\left( {{L_u}} \right)$ can be characterized by solving problem \eqref{C1} to obtain its optimal value ${R^*}\left( {\bm{\alpha }} \right)$ with different ${\bm{\alpha }}$'s.
\begin{figure}[!t]
\setlength{\abovecaptionskip}{-5pt}
\setlength{\belowcaptionskip}{-5pt}
\centering
\includegraphics[width= 0.55\textwidth]{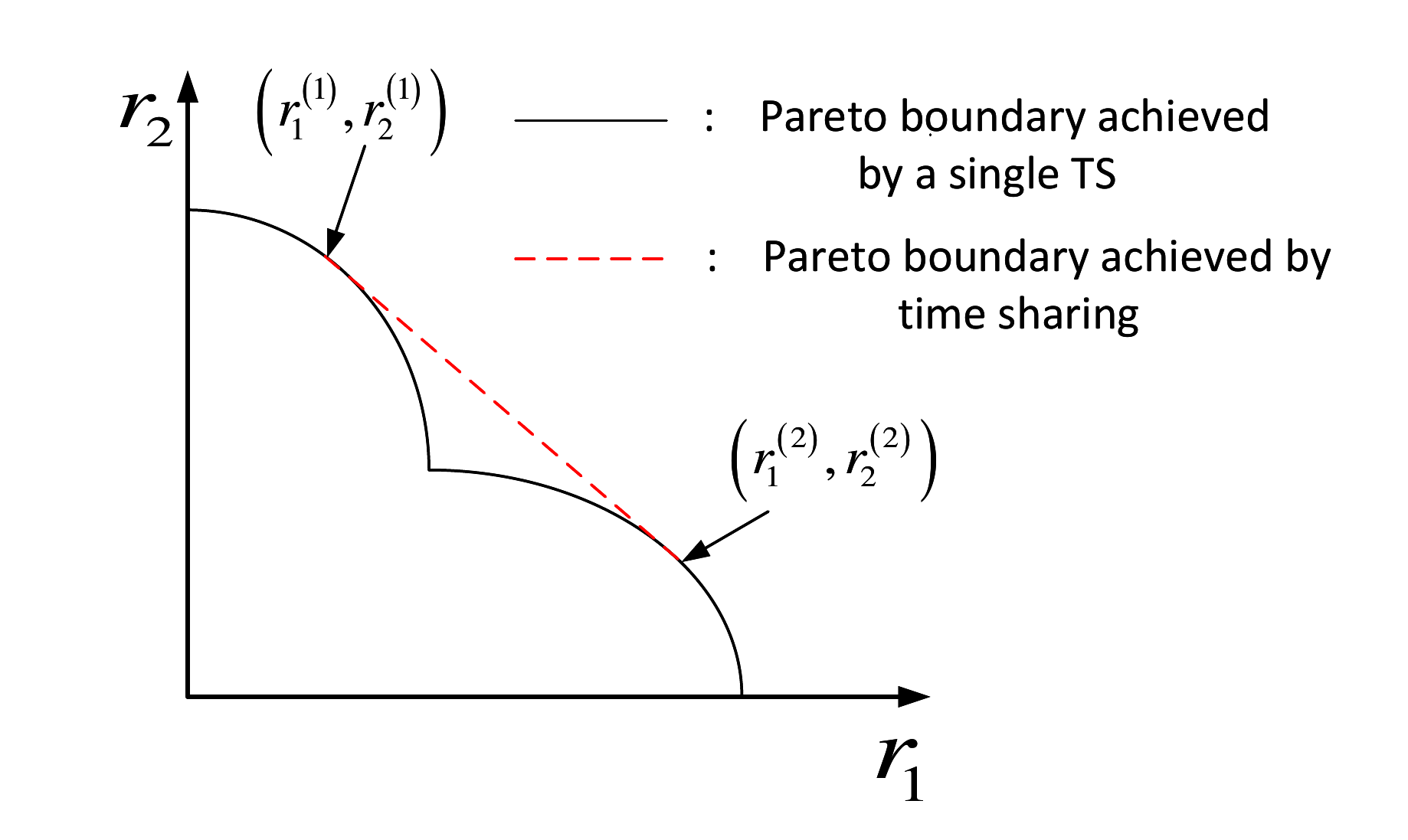}
\DeclareGraphicsExtensions.
\caption{Illustration of the capacity region.}
\label{illustration}
\vspace{-6pt}
\end{figure}
\begin{rem}
We provide the graphical illustration for Proposition 1 through a two-user case in Fig. \ref{illustration}. The black solid line in Fig. \ref{illustration} is obtained by optimally solving problem \eqref{C1} under different ${\bm{\alpha }}$. Suppose that $\left( {r_1^{\left( 1 \right)},r_2^{\left( 1 \right)}} \right)$ and $\left( {r_1^{\left( 2 \right)},r_2^{\left( 2 \right)}} \right)$ are achieved by using ${\Xi ^{\left( 1 \right)}} = \left\{ {{{\bf{\Theta }}^{\left( 1 \right)}},{\bf{S}}_k^{\left( 1 \right)}} \right\}$ and ${\Xi ^{\left( 2 \right)}} = \left\{ {{{\bf{\Theta }}^{\left( 2 \right)}},{\bf{S}}_k^{\left( 2 \right)}} \right\}$, respectively. Any point on the red dash line can be achieved by performing time sharing among ${\Xi ^{\left( 1 \right)}}$ and ${\Xi ^{\left( 2 \right)}}$. Specifically, by allocating $\beta T$ for ${\Xi ^{\left( 1 \right)}}$ and $\left( {1 - \beta } \right)T$ for ${\Xi ^{\left( 2 \right)}}$, the rate pair $\left( {r_1^{\left( 1 \right)} + \beta \left( {r_1^{\left( 2 \right)} - r_1^{\left( 1 \right)}} \right),r_2^{\left( 1 \right)} - \beta \left( {r_2^{\left( 1 \right)} - r_2^{\left( 2 \right)}} \right)} \right)$ can be achieved, where $0 \le \beta  \le 1$. Therefore, the capacity region can be further enlarged if the region obtained by the single TS is non-convex. It is worth noting that performing time sharing among ${\Xi ^{\left( 1 \right)}}$ and ${\Xi ^{\left( 2 \right)}}$ requires dynamic beamforming configurations, which incurs more hardware cost and signalling overhead.
\end{rem}
\vspace{-8pt}
\subsection{Exact and Inner Bounds of Capacity Region}
In this subsection, we focus on problem \eqref{C1} to characterize the capacity region. Note that even in the absence of the IRS, the left-hand side (LHS) of \eqref{C1-b} is non-concave with respect to $\left\{ {{{\bf{S}}_k}} \right\}$, which makes \eqref{C1} a non-convex optimization problem. Moreover, the introduction of the IRS renders problem \eqref{C1} more challenging since ${\bf{\Theta }}$ and $\left\{ {{{\bf{S}}_k}} \right\}$ are tightly coupled in constraint \eqref{C1-b}. Nevertheless, we propose the optimal and high-quality suboptimal solutions to problem \eqref{C1}, based on which both the exact and inner bounds of the capacity region can be derived, respectively.

Note that the optimal objective value of problem \eqref{C1} is a function with respect to ${P_{\max }}$, i.e., ${R^*}\left( {\bm{\alpha }} \right) = f\left( {{P_{\max }}} \right)$. It can be readily shown that $f\left( {{P_{\max }}} \right)$ increases with respect to ${P_{\max }}$. As such, for a given target rate ${R^t}$, the minimum transmit power required at the BS is denoted by ${P^t} = {f^{ - 1}}\left( {{R^t}} \right)$ and ${f^{ - 1}}\left( {{R^t}} \right)$ increases with respect to ${{R^t}}$. Therefore, ${R^*}\left( {\bm{\alpha }} \right)$ is the unique root of the equation ${f^{ - 1}}\left( {{R^t}} \right) = {P_{\max }}$ with respect to ${{R^t}}$, which can be obtained by a bisection search. In the following, we aim to characterize the required transmit power for supporting any given target rate by considering the following power minimization problem, i.e.,
\begin{subequations}\label{C2}
\begin{align}
\label{C2-a}\mathop {\min }\limits_{{\bf{\Theta }},\left\{ {{{\bf{S}}_k}} \right\},\Omega}  \;\;&\sum\nolimits_{k = 1}^K {{\mathop{\rm Tr}\nolimits} } \left( {{{\bf{S}}_k}} \right)\\
\label{C2-b}{\rm{s.t.}}\;\;\;\;\;&{\gamma _k}\left( {{\bf{h}}_k^H\left( {\bf{\Theta }} \right)\left( {\sum\limits_{i > k}^K {{{\bf{S}}_i}} } \right){{\bf{h}}_k}\left( {\bf{\Theta }} \right) + {\sigma ^2}} \right) \le {\bf{h}}_k^H\left( {\bf{\Theta }} \right){{\bf{S}}_k}{{\bf{h}}_k}\left( {\bf{\Theta }} \right), ~\forall {k} \in {{\cal K}},\\
\label{C2-e}&\eqref{C1-d},~\eqref{C1-e},
\end{align}
\end{subequations}
where ${\gamma _k} = {2^{{\alpha _k}{R^t}}} - 1$.

For problem \eqref{C2}, $\left\{ {{{\bf{S}}_k}} \right\}$ and ${\bf{\Theta }}$ are still coupled in constraint \eqref{C2-b}. To overcome this issue, the functional relationship between the objective value of problem \eqref{C2} and ${\bf{\Theta }}$ is derived in the following proposition by fully exploiting the inherent structure of problem \eqref{C2}.
\begin{pos}
The optimal value of problem \eqref{C2} can be expressed in closed-form as
\begin{align}\label{power_value}
{p^*}\left( {{\bf{\Theta }}} \right) = \sum\nolimits_{k = 1}^K {{\lambda _k}\left( {{\bf{\Theta }} } \right)} {\gamma _k},
\end{align}
where
\begin{align}\label{lamda_group1}
&{\lambda _k}\left( {{\bf{\Theta }}} \right) \!\!=\!\! \frac{{{\sigma ^2}}}{{{\bf{h}}_k^H\left( {\bf{\Theta }} \right){{\left( {{{\bf{I}}_M} \!\!+\!\! {{\bf{A}}_k}\left( {{\bf{\Theta }}} \right)} \right)}^{ - 1}}{{\bf{h}}_k}\left( {\bf{\Theta }} \right)}},k > 1,\\
&{\lambda _k}\left( {{\bf{\Theta }}} \right) = \frac{{{\sigma ^2}{\gamma _k}}}{{{\bf{h}}_k^H\left( {\bf{\Theta }} \right){{\bf{h}}_k}\left( {\bf{\Theta }} \right)}},k= 1,
\end{align}
with
\begin{align}\label{power_value2}
{{\bf{A}}_k}\left( {{\bf{\Theta }}} \right) = \sum\nolimits_{i > k} {\frac{{{\gamma _k}{\lambda _i}\left( {\bf{\Theta }} \right)}}{{{\sigma ^2}}}{{\bf{h}}_i}\left( {\bf{\Theta }} \right){\bf{h}}_i^H\left( {\bf{\Theta }} \right)} ,k > 1.
\end{align}
\end{pos}
\begin{proof}
Please refer to Appendix A.
\end{proof}

Proposition 2 identifies the fact that the role of the IRS in the MISO broadcast channel is to minimize the cost function given in \eqref{power_value}. By exploiting Proposition 2, both the exact and inner bounds of the IRS aided MISO broadcast channel can be characterized as follows.
\subsubsection{Exact Bound Characterization}
Due to the non-convex function ${p^*}\left( {{\bf{\Theta }}} \right)$ and the discrete unit-modulus constraint on ${\bf{\Theta }}$, there is generally no standard approach for efficiently solving the problem $\mathop {\min }\limits_{{\bf{\Theta }}} {p^*}\left( {{\bf{\Theta }}} \right)$ optimally. One straightforward method to obtain the optimal solution is to adopt exhaustive search. By performing exhaustive searching for \eqref{power_value}, the minimum transmit power required at the BS under any given rate target can be obtained as
\begin{align}\label{power_function1}
{f^{ - 1}}\left( {{R^t}} \right) = \mathop {\min }\limits_{{{\left[ {\bf{\Theta }} \right]}_{n,n}} \in {\cal F}} {p^*}\left( {\bf{\Theta }} \right).
\end{align}
Then, we further employ bisection searching to obtain the root of the equation ${f^{ - 1}}\left( {{R^t}} \right) = {P_{\max }}$ and thus the optimal value of problem \eqref{C1} is obtained as ${R^*}\left( {\bm{\alpha }} \right)$. The procedure of obtaining ${R^*}\left( {\bm{\alpha }} \right)$ is summarized in Algorithm 1. The computational complexity of Algorithm 1 is dominated by step 4. Specifically, the complexity of performing exhaustive search in step 4 is ${\cal O}\left( {{N^Q}} \right)$. Hence, the total complexity for obtaining the objective value of problem \eqref{C3} is ${\cal O}\left( {{{\log }_2}\left( {{R_{\max }}/\varepsilon } \right){N^Q}} \right)$, where ${{{\log }_2}\left( {{R_{\max }}/\varepsilon } \right)}$ represents the maximum number of iterations for bisection search.

\begin{algorithm}[!t]
 \caption{Bisection-based algorithm for obtaining ${R^*}\left( {\bm{\alpha }} \right)$.}
 \label{alg1}
 \begin{algorithmic}[1]
  \STATE  \textbf{Initialize} threshold $\varepsilon  > 0$, ${R_{\min }} = 0$, and ${R_{\max }} \ge {R^*}\left( {\bm{\alpha }} \right)$.
  \STATE  \textbf{repeat}
  \STATE \quad Set $\bar R = \frac{{{R_{\min }} + {R_{\max }}}}{2}$.
  \STATE \quad Solve problem \eqref{power_function1} by using exhaustive searching to obtain ${f^{ - 1}}\left( {{R^t}} \right)$. If ${f^{ - 1}}\left( {{R^t}} \right) > {P_{\max }}$,\\
\quad set ${R_{\max }} = \bar R$. Otherwise, set ${R_{\min }} = \bar R$.
  \STATE \textbf{until} ${R_{\max }} - {R_{\min }} \le \varepsilon$, where $\varepsilon  > 0$ is a small constant to control the algorithm accuracy.
  \STATE \textbf{Output} ${R^*}\left( {\bf{\alpha }} \right) = \bar R$.
 \end{algorithmic}
\end{algorithm}

By varying the values of ${\bm{\alpha }}$,  the resulting rate tuples achieved by using a single TS can be obtained, which constitute the Pareto boundary of the capacity region under the static beamforming configuration. Finally, the exact bound of the capacity region achieved by using dynamic beamforming configuration can be characterized based on \eqref{region_construction} in Proposition 1, which is given by ${{\cal C}^{\rm{O}}} = {\mathop{\rm Conv}\nolimits} \left( {{{\bf{0}}_{K \times 1}}\bigcup\nolimits_{{\bm{\alpha }},{\bm{\alpha }} \in {\cal A}} {{R^*}\left( {\bm{\alpha }} \right){\bm{\alpha }}} } \right)$. It is worth noting that ${{\cal C}^{\rm{O}}}$ serves as the achievable capacity upper bound, which is useful to measure the practical performance gap between the easy-implementation scheme and theoretical capacity limit.

\subsubsection{Inner Bound Characterization}
To reduce the computational complexity incurred by searching the optimal ${\bf{\Theta }}$ when $N$ is practically large. A low complexity algorithm based on element-wise alternating optimization can be employed for solving the problem $\mathop {\min }\limits_{{\bf{\Theta }}} {p^*}\left( {{\bf{\Theta }}} \right)$. Specifically, by fixing the phase-shifts of the set ${\cal N}\backslash n$, the optimal solution of the phase-shift the $n$-th element can be obtained by an one-dimensional search, i.e.,
\begin{align}\label{optimal_phase}
\theta _n^* = \arg \mathop {\min }\limits_{{\theta _n} \in {\cal F}} {p^*}\left( {{\bf{\Theta }}} \right).
\end{align}
By alternately determining the phase-shifts of all elements based on \eqref{optimal_phase}, the value of ${p^*}\left( {{\bf{\Theta }}} \right)$ is non-decreasing over the iterations. Moreover, the minimum value of ${p^*}\left( {{\bf{\Theta }}} \right)$ is lower-bounded by a finite value and thus the convergence is guaranteed. Denote the converged value of ${p^*}\left( {{\bf{\Theta }}} \right)$ by ${{\tilde f}^{ - 1}}\left( {{R^t}} \right)$. Then, the root of the equation ${{\tilde f}^{ - 1}}\left( {{R^t}} \right) = {P_{\max }}$, denoted by $\tilde R\left( {\bm{\alpha }} \right)$, is obtained by using bisection searching. The procedure of obtaining $\tilde R\left( {\bm{\alpha }} \right)$ is summarized in Algorithm 2. The complexity of Algorithm 2 can be analyzed as follows. In the inner layer to minimize ${p^*}\left( {{\bf{\Theta }}} \right)$, the resulting computational complexity is ${\cal O}\left( {{I_{{\rm{iter}}}}N} \right)$, where ${{I_{{\rm{iter}}}}}$ denotes the inner iteration numbers associated with using element-wise alternating method to update ${\theta _n}$. In the outer layer, the bisection search is employed and thus the overall complexity of Algorithm 2 can be written as ${\cal O}\left( {{{\log }_2}\left( {{R_{\max }}/\varepsilon } \right){I_{{\rm{iter}}}}N} \right)$.

\begin{algorithm}[!t]
 \caption{Element wise alternating-based algorithm for obtaining $\tilde R\left( {\bm{\alpha }} \right)$.}
 \label{alg1}
 \begin{algorithmic}[1]
  \STATE  \textbf{Initialize} ${\varepsilon _1} > 0$, ${\varepsilon _2} > 0$, ${R_{\min }} = 0$, and ${R_{\max }} \ge {R^*}\left( {\bm{\alpha }} \right)$
  \STATE  \textbf{repeat}
  \STATE \quad Set $\bar R = \frac{{{R_{\min }} + {R_{\max }}}}{2}$.
  \STATE \quad \textbf{repeat}
  \STATE \quad \quad \textbf{for} $n = 1 \to N$ \textbf{do}
  \STATE \quad \quad \quad Update ${\theta _n}$ as \eqref{optimal_phase}.
  \STATE \quad \quad \textbf{end for}
   \STATE \quad \textbf{until} the fractional decrease of the objective value of ${p^*}\left( {{\bf{\Theta }}} \right)$ is below ${\varepsilon _1} > 0$.
  \STATE \quad If ${p^*}\left( {{\bf{\Theta }}} \right) > {P_{\max }}$,
               set ${R_{\max }} = \bar R$. Otherwise, set ${R_{\min }} = \bar R$.
  \STATE \textbf{until} ${R_{\max }} - {R_{\min }} \le {\varepsilon _2}$.
  \STATE \textbf{Output} $\tilde R\left( {\bm{\alpha }} \right) = \bar R$.
 \end{algorithmic}
\end{algorithm}

Similarly, the inner bound of the capacity region can be characterized as
\begin{align}\label{Inner_bound}
{{\cal C}^{\rm{I}}} = {\mathop{\rm Conv}\nolimits} \left( {{{\bf{0}}_{K \times 1}}\bigcup\nolimits_{{\bm{\alpha }},{\bm{\alpha }} \in {\cal A}} {\tilde R\left( {\bm{\alpha }} \right){\bm{\alpha }}} } \right).
\end{align}
Compared to the characterization of the exact bound, the complexity of characterizing the inner bound is linear over $N$ and thus is significantly reduced. The tightness of the derived inner bound over the exact bound will be verified by simulations under different channel setups in Section V.

\begin{rem}
Note that the developed bisection search based framework in this section is also applicable to characterize the achievable rate region of the IRS aided ZF scheme. For any given ${\bm{\alpha }}$, the minimum required transmit power ${p^{{\rm{zf}}}}\left( {{R^t}} \right)$ for supporting a given target rate ${R^t}$ can be obtained similar to \eqref{power_function1}. Then, the unique root of the equation ${p^{{\rm{zf}}}}\left( {{R^t}} \right) = {P_{\max }}$ is obtained by $R_{{\rm{zf}}}^*\left( {\bm{\alpha }} \right)$. The achievable rate region of the IRS aided ZF scheme can be characterized as ${{\cal R}^{{\rm{zf}}}} = {\mathop{\rm Conv}\nolimits} \left( {{{\bf{0}}_{K \times 1}}\bigcup\nolimits_{{\bm{\alpha }},{\bm{\alpha }} \in {\cal A}} {R_{{\rm{zf}}}^*\left( {\bm{\alpha }} \right){\bm{\alpha }}} } \right)$.
\end{rem}

\section{Theoretical Analysis}
The integration of the IRS facilitates the flexibility of smartly controlling multiuser wireless channels, which may have different impacts on DPC and ZF-based schemes in Section II. In this section, we aim to answer two fundamental questions from the theoretical analysis perspective. The first one is whether dynamic beamforming configurations bring additional gains or not? The second one is whether ZF based transmission scheme is capable of approaching the performance of the DPC scheme or not after integrating IRS?

\subsection{How IRS Affects Channel Correlations}
Note that the answers to the above two questions rely heavily on the channel correlation, which is essentially the function of the IRS phase-shift. Before answering the aforementioned two questions, we first give some discussions on how IRS affects the channel correlation. Define the channel correlation between two users (user $k$ and user $m$) as follows
\begin{align}\label{correlation}
{\rho _{k,m}} = \frac{{\left| {{\bf{h}}_k^H\left( {\bf{\Theta }} \right){{\bf{h}}_m}\left( {\bf{\Theta }} \right)} \right|}}{{\left\| {{{\bf{h}}_m}\left( {\bf{\Theta }} \right)} \right\|\left\| {{{\bf{h}}_k}\left( {\bf{\Theta }} \right)} \right\|}},
\end{align}
which satisfies $0 \le {\rho _{k,m}} \le 1$.
\begin{lem}
Under the case that ${\rho _{k,m}} = 0$, i.e., $\left( {{\bf{h}}_{r,k}^H{\bf{\Theta G}} + {\bf{h}}_{d,k}^H} \right){\left( {{\bf{h}}_{r,m}^H{\bf{\Theta G}} + {\bf{h}}_{d,m}^H} \right)^H} = 0$,~$\forall k \ne m,k,m \in {\cal K}$, we always have ${{\cal R}^{{\rm{zf}}}} = {{\cal C}^{\rm{O}}}$.
\end{lem}
\begin{proof}
To show ${{\cal R}^{{\rm{zf}}}} = {{\cal C}^{\rm{O}}}$, we only need to prove $R_{{\rm{zf}}}^*\left( {\bm{\alpha }} \right) = {R^*}\left( {\bm{\alpha }} \right)$. It is obvious that $R_{{\rm{zf}}}^*\left( {\bm{\alpha }} \right) \le {R^*}\left( {\bm{\alpha }} \right)$ holds naturally since DPC is capacity-achieving. Then, it can be shown that ${R^*}\left( {\bm{\alpha }} \right) \le R_{{\rm{zf}}}^*\left( {\bm{\alpha }} \right)$ under the case of ${\rho _{k,m}} = 0,\forall k \ne m$. By removing the inter-user interference, the sum-rate of the DPC scheme under arbitrarily given power allocation, i.e., $\left\{ {{p_k}} \right\}$, is upper bounded by
\begin{align}\label{DPC_upperbound}
{R^*}\left( {\bm{\alpha }} \right) \le \sum\limits_{k = 1}^K {{{\log }_2}\left( {1 + \frac{{{p_k}{{\left\| {{{\bf{h}}_k}\left( {\bf{\Theta }} \right)} \right\|}^2}}}{{{\sigma ^2}}}} \right)}.
\end{align}
By setting ${\bf{w}}_k^* = {p_k}{{\bf{h}}_k}\left( {\bf{\Theta }} \right)/\left\| {{{\bf{h}}_k}\left( {\bf{\Theta }} \right)} \right\|$, it can be readily verified that the RHS of \eqref{DPC_upperbound} can be achieved by ZF-based scheme since ${\left| {{{\bf{h}}_m}{\bf{w}}_k^*} \right|^2} = 0$ for $k \ne m$ and ${\left| {{{\bf{h}}_k}{\bf{w}}_k^*} \right|^2} = {\left\| {{{\bf{h}}_k}\left( {\bf{\Theta }} \right)} \right\|^2}$. Hence, we have ${R^*}\left( {\bm{\alpha }} \right) \le R_{{\rm{zf}}}^*\left( {\bm{\alpha }} \right)$ in this case. Given $R_{{\rm{zf}}}^*\left( {\bm{\alpha }} \right) \le {R^*}\left( {\bm{\alpha }} \right)$ and ${R^*}\left( {\bm{\alpha }} \right) \le R_{{\rm{zf}}}^*\left( {\bm{\alpha }} \right)$, we have ${R^*}\left( {\bm{\alpha }} \right) = R_{{\rm{zf}}}^*\left( {\bm{\alpha }} \right)$.
\end{proof}

Lemma 1 indicates that under the specific condition of ${\rho _{k,m}} = 0,\forall k \ne m$, i.e., the effective channels of all users are orthogonal, the achievable rate region of ZF is the same as the capacity region of DPC. For the conventional massive MIMO system, it has been shown that the favourable orthogonal channel condition can be satisfied as the number of antennas at the BS tends to be infinite, i.e., $M \to \infty$, which however, is rarely satisfied for small $M$. In the considered IRS aided MISO broadcast channel, it is interesting to see whether the orthogonal channels can be obtained or not by employing a large number of IRS elements instead of transmit antennas at the BS. Under the condition of randomly given IRS phase-shifts, we have the following lemma.

\begin{lem}
Assume that ${{\bf{h}}_{r,k}}$'s, ${\bf{G}}$ are independent and identically distributed variables with ${{\bf{h}}_{r,k}} \sim  {\cal C}{\cal N}\left( {0,{\bf{I}}} \right)$,$\forall k \in {\cal K}$, and ${\bf{G}} \sim {\cal C}{\cal N}\left( {0,{\bf{I}}} \right)$. Under randomly given ${\bf{\Theta }}$, we have
\begin{align}\label{correlation1_mean}
\frac{{{\mathop{\mathbb{E}}\nolimits} \left\{ {{{\left| {{\bf{h}}_{r,k}^H{\bf{\Theta G}}{{\bf{G}}^H}{{\bf{\Theta }}^H}{{\bf{h}}_{r,m}}} \right|}^2}} \right\}}}{{{\mathop{\mathbb{E}}\nolimits} \left\{ {{{\left\| {{\bf{h}}_{r,k}^H{\bf{\Theta G}}} \right\|}^2}} \right\}{\mathop{\mathbb{E}}\nolimits} \left\{ {{{\left\| {{\bf{h}}_{r,m}^H{\bf{\Theta G}}} \right\|}^2}} \right\}}} \ge \frac{1}{M},\forall k \ne m.
\end{align}
\end{lem}
\begin{proof}
Please refer to Appendix B.
\end{proof}

In the scenario that the direct links are blocked, Lemma 2 demonstrates that the channel correlation of the two users with random IRS phase-shifts is always lower bounded by a specific value, which depends only on the number of antennas at the BS rather than the number of IRS elements. However, this result is under the condition that the IRS phase-shifts are not properly optimized, which means that the potential capability of the IRS in reducing the channel correlation is not fully unleashed. Thus, it is not clear whether the square of the channel correlation can be further reduced to a value which is far less than $1/M$ by carefully designing the IRS phase-shifts. To this end, we consider the optimization problem to minimize the maximum correlation coefficient among users as follows.
\begin{align}\label{correlation_minimization}
\mathop {\min }\limits_{\bf{\Theta }} \mathop {\max }\limits_{k,m \in {\cal K},k \ne m} \frac{{{{\left| {\left( {{\bf{h}}_{d,k}^H + {\bf{h}}_{r,k}^H{\bf{\Theta G}}} \right){{\left( {{\bf{h}}_{d,m}^H + {\bf{h}}_{r,m}^H{\bf{\Theta G}}} \right)}^H}} \right|}^2}}}{{{{\left\| {{\bf{h}}_{d,k}^H + {\bf{h}}_{r,k}^H{\bf{\Theta G}}} \right\|}^2}{{\left\| {{\bf{h}}_{d,m}^H + {\bf{h}}_{r,m}^H{\bf{\Theta G}}} \right\|}^2}}} \;\;\;\;{\rm{s}}.{\rm{t}}.\;\;{\left[ {\bf{\Theta }} \right]_{n,n}} \in {\cal F},~\forall n.
\end{align}
The local-optimal solution of problem \eqref{correlation_minimization} can be efficiently obtained by using the element-wise alternating method introduced in the previous section.

{\bf{Numerical Example:}} To evaluate the capability of the IRS in reducing the channel correlations among users, the following numerical example is provided. We consider the typical setup that direct links between the BS and users are blocked. Additionally, we set $M=4$, $K=2$, $\bar N =1$ and all the channels follow Rayleigh distribution. Since $\bar N =1$, $N$ represents the number of IRS elements in this case.

\begin{figure}[t!]
\centering
\includegraphics[width = 3in]{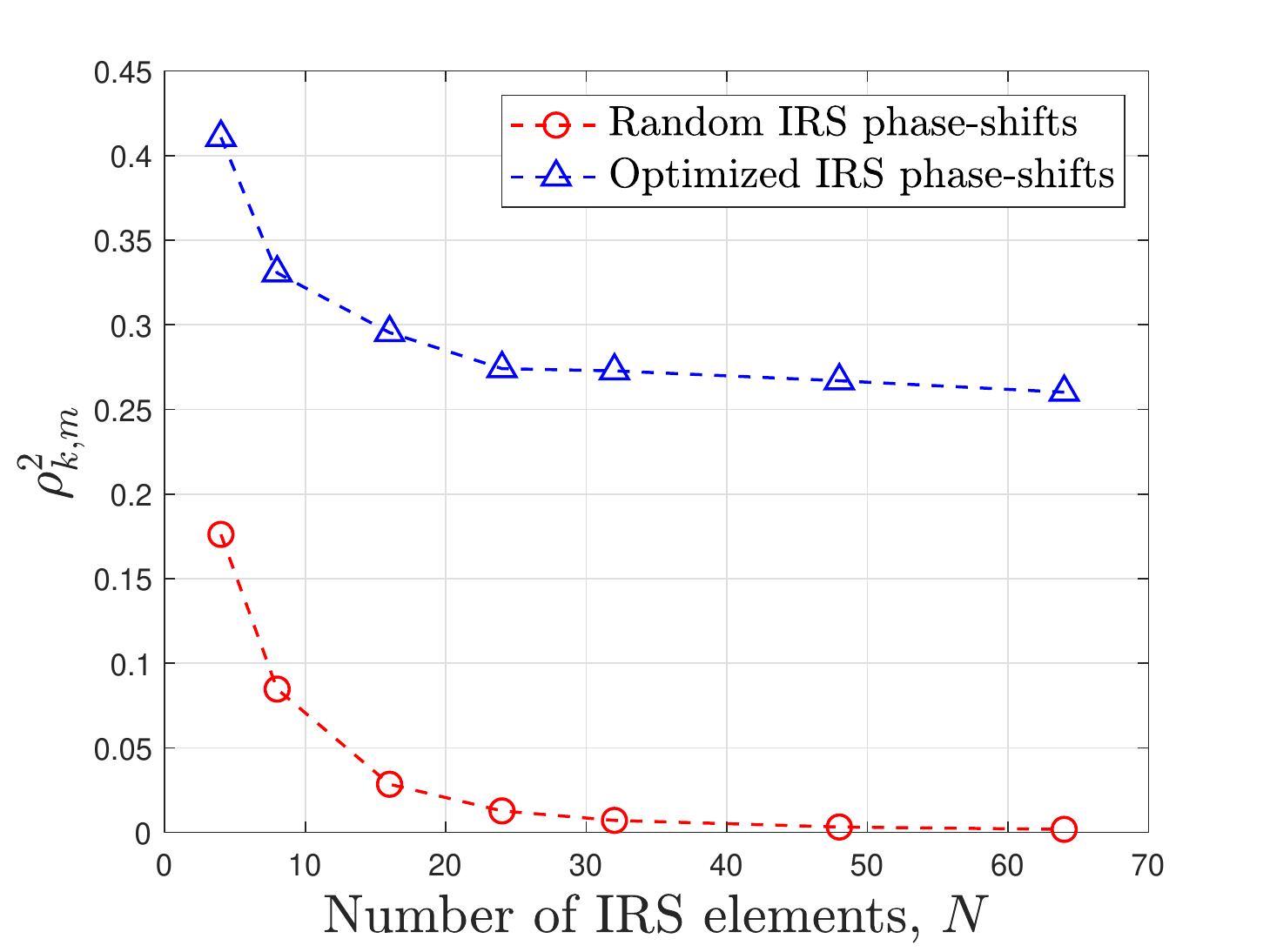}
\caption{{The square of channel correlation coefficient versus $N$.}}
\label{orthogonality_eva}
\vspace{-16pt}
\end{figure}

In Fig. \ref{orthogonality_eva}, we plot the square of channel correlation coefficient, i.e., $\rho _{k,m}^2$, versus the number of IRS reflecting elements. In the case of random IRS phase-shifts, it is observed that $\rho _{k,m}^2$ is lower bounded by $1/M = 1/4$ as $N$ increases, which is consistent with the analytical result in Lemma 2. In contrast, under the case of optimized IRS phase-shifts, $\rho _{k,m}^2$ is observed to significantly decrease with $N$ as expected. In particular, $\rho _{k,m}^2$ is capable of approaching zero (less than ${10^{ - 3}}$) when $N>60$, which demonstrates the potential capability of the IRS in achieving the favourable propagation condition.

\begin{rem}
The rate loss of ZF compared to DPC highly depends on the channel correlation coefficient among users. Under the setup of the two-user case, when ${\left\| {{{\bf{h}}_k}\left( {\bf{\Theta }} \right)} \right\|^2}/{\sigma ^2} \gg 1$, the corresponding rate loss for user $k$ can be expressed as
\begin{align}\label{loss}
R_k^{{\rm{DPC}}} - R_k^{ZF} &\le {\log _2}\left( {1 + \frac{{{p_k}{{\left\| {{{\bf{h}}_k}\left( {\bf{\Theta }} \right)} \right\|}^2}}}{{{\sigma ^2}}}} \right) - {\log _2}\left( {1 + \frac{{{p_k}\left( {1 - \rho _{k,m}^2} \right){{\left\| {{{\bf{h}}_k}\left( {\bf{\Theta }} \right)} \right\|}^2}}}{{{\sigma ^2}}}} \right)\nonumber\\
& \cong  - {\log _2}\left( {1 - \rho _{k,m}^2} \right),
\end{align}
where $R_k^{{\rm{DPC}}}$ and $R_k^{{\rm{ZF}}}$ denote the rate of user $k$ under the scheme of DPC and ZF, respectively. Note that the corresponding rate loss is less than $1.4 \times {10^{ - 3}}$ (bps/Hz) under the condition that $\rho _{k,m}^2 \le {10^{ - 3}}$, which is negligible in practical wireless systems.
\end{rem}

\subsection{Is Dynamic BS/IRS BF Configuration Needed?}
In this subsection, we investigate whether dynamic beamforming configurations are beneficial for improving the achievable rate. Regarding the ZF based transmission scheme, we provide the following proposition to unveil the effectiveness of dynamic beamforming configurations\footnote{Regarding the IRS aided DPC, numerical results are provided in Section V under different channel setups to shed light on whether dynamic beamforming can further enlarge its capacity region or not.}. In the setup of the two-user case, the Pareto boundaries of the rate pairs achieved by static and dynamic beamforming under the ZF-based transmission scheme are denoted by ${\cal R}_{{\rm{ZF}}}^{\rm{s}}$ and ${\cal R}_{{\rm{ZF}}}^{\rm{d}}$, respectively. Then, we have the following proposition.
\begin{pos}
If $\left| {{\bf{h}}_1^H\left( {\bf{\Theta }} \right){{\bf{h}}_2}\left( {\bf{\Theta }} \right)} \right|/\left( {{{\left\| {{{\bf{h}}_1}\left( {\bf{\Theta }} \right)} \right\|}}{{\left\| {{{\bf{h}}_2}\left( {\bf{\Theta }} \right)} \right\|}}} \right) > 0$, $\forall {\left[ {\bf{\Theta }} \right]_{n,n}} \in {\cal F}$, there always exist rate pairs $\left( {\bar r_1^{{\rm{ZF}},s},\bar r_2^{{\rm{ZF}},s}} \right) \in {\cal R}_{{\rm{ZF}}}^{\rm{s}}$ and $\left( {\bar r_1^{{\rm{ZF}},{\rm{d}}},\bar r_2^{{\rm{ZF}},{\rm{d}}}} \right) \in {\cal R}_{{\rm{ZF}}}^{\rm{d}}$, which satisfy
\begin{align}\label{rate_pair_comparison}
\left( {\bar r_1^{{\rm{ZF,d}}},\bar r_2^{{\rm{ZF,d}}}} \right) \succ \left( {\bar r_1^{{\rm{ZF}},s},\bar r_2^{{\rm{ZF}},s}} \right).
\end{align}
\end{pos}
\begin{proof}
For the given rate profile vector ${\bm{\alpha }} = {\left[ {0,1} \right]^T}$, the corresponding rate pair is achieved by maximizing ${\left\| {{\bf{h}}_{d,2}^H + {\bf{h}}_{r,2}^H{\bf{\Theta G}}} \right\|^2}$ and its associated optimal IRS phase-shift matrix is denoted by ${{\bf{\Theta }}^*}$. By employing ${{\bf{w}}_1} = {\bf{0}}$, ${{\bf{w}}_2} = \sqrt {{P_{\max }}} {{\bf{h}}_2}\left( {{{\bf{\Theta }}^*}} \right)/\left\| {{{\bf{h}}_2}\left( {{{\bf{\Theta }}^*}} \right)} \right\|$, and ${\bf{\Theta }} = {{\bf{\Theta }}^*}$, the corresponding rate pair, denoted by $\left( {0,\hat r_2^{\max }} \right)$, can be obtained, where $\hat r_2^{\max } = {\log _2}\left( {1 + {{{P_{\max }}{{\left\| {{{\bf{h}}_2}\left( {{{\bf{\Theta }}^*}} \right)} \right\|}^2}} \mathord{\left/
 {\vphantom {{{P_{\max }}{{\left\| {{{\bf{h}}_2}\left( {{{\bf{\Theta }}^*}} \right)} \right\|}^2}} {{\sigma ^2}}}} \right.
 \kern-\nulldelimiterspace} {{\sigma ^2}}}} \right)$. For any rate profile vector ${\bm{\alpha }}$ with $0 < {\bm{\alpha }}\left[ 1 \right] < 1$, the achievable rate of user 2 under static beamforming  configuration is upper-bounded by
\begin{align}\label{rate2_upperbound}
r_2^{\rm{s}} &= {\log _2}\left( {1 + \frac{{{{\left\| {{\bf{h}}_2^H\left( {\bf{\Theta }} \right){{\bf{w}}_1}} \right\|}^2}}}{{{\sigma ^2}}}} \right)\nonumber\\
& < {\log _2}\left( {1 + \frac{{{P_{\max }}{{\left\| {{{\bf{h}}_2}\left( {{{\bf{\Theta }}^*}} \right)} \right\|}^2}}}{{{\sigma ^2}}}\left( {1 - \frac{{{{\left| {{\bf{h}}_1^H\left( {{\bf{\tilde \Theta }}} \right){{\bf{h}}_2}\left( {{\bf{\tilde \Theta }}} \right)} \right|}^2}}}{{{{\left\| {{{\bf{h}}_1}\left( {{\bf{\tilde \Theta }}} \right)} \right\|}^2}{{\left\| {{{\bf{h}}_2}\left( {{\bf{\tilde \Theta }}} \right)} \right\|}^2}}}} \right)} \right) \buildrel \Delta \over = r_2^{{\rm{s,ub}}},
\end{align}
where ${{\bf{\tilde \Theta }}}$ is the optimal phase-shift for minimizing $\left| {{\bf{h}}_1^H\left( {\bf{\Theta }} \right){{\bf{h}}_2}\left( {\bf{\Theta }} \right)} \right|/\left( {\left\| {{{\bf{h}}_1}\left( {\bf{\Theta }} \right)} \right\|\left\| {{{\bf{h}}_2}\left( {\bf{\Theta }} \right)} \right\|} \right)$. Since $\left| {{\bf{h}}_1^H\left( {\bf{\Theta }} \right){{\bf{h}}_2}\left( {\bf{\Theta }} \right)} \right|/\left( {{{\left\| {{{\bf{h}}_1}\left( {\bf{\Theta }} \right)} \right\|}}{{\left\| {{{\bf{h}}_2}\left( {\bf{\Theta }} \right)} \right\|}}} \right) > 0$, $\forall {\left[ {\bf{\Theta }} \right]_{n,n}} \in {\cal F}$, we have $r_2^{{\rm{s,ub}}} < \hat r_2^{\max }$. For arbitrarily chosen rate pair $\left( {r_1^A,r_2^A} \right) \in {\cal R}_{{\rm{ZF}}}^{\rm{s}}$, its associated BS/IRS beamforming configuration is denoted by $\left\{ {{\bf{w}}_1^A,{\bf{w}}_2^A,{{\bf{\Theta }}^A}} \right\}$. Note that $0 < r_2^A < r_2^{{\rm{s,ub}}}$ and $r_1^A > 0$, which thus leading to $0 < \left( {r_2^{\max } - r_2^{{\rm{s,ub}}}} \right)/\left( {r_2^{\max } - r_2^A} \right) < 1$. Then, by allocating $\left( {r_2^{\max } - r_2^{{\rm{s,ub}}}} \right)T/\left( {r_2^{\max } - r_2^A} \right)$ for $\left\{ {{\bf{w}}_1^A,{\bf{w}}_2^A,{{\bf{\Theta }}^A}} \right\}$ and $\left( {r_2^{{\rm{s,ub}}} - r_2^A} \right)T/\left( {r_2^{\max } - r_2^A} \right)$ for $\left\{ {{\bf{0}},\sqrt {{P_{\max }}} {{\bf{h}}_2}\left( {{{\bf{\Theta }}^*}} \right)/\left\| {{{\bf{h}}_2}\left( {{{\bf{\Theta }}^*}} \right)} \right\|,{{\bf{\Theta }}^*}} \right\}$, the rate pair $\left( {\hat r_1^B,\hat r_2^B} \right)$ can be achieved with
\begin{align}\label{rate_pair}
\hat r_1^B = \frac{{r_2^{\max } - r_2^{{\rm{s,ub}}}}}{{r_2^{\max } - r_2^A}}r_1^A,\hat r_2^B = r_2^{{\rm{s,ub}}}.
\end{align}
For the rate pair $\left( {\hat r_1^B,r_2^B} \right) \in {\cal R}_{{\rm{ZF}}}^{\rm{s}}$, $r_2^B < r_2^{{\rm{s,ub}}} = \hat r_2^B$ holds since \eqref{rate2_upperbound}. Hence, we have $\left( {\hat r_1^B,r_2^B} \right) \prec \left( {\hat r_1^B,\hat r_2^B} \right)$. Moreover, it is obvious that there always exists one rate pair $\left( {\bar r_1^{{\rm{ZF}},{\rm{d}}},\bar r_2^{{\rm{ZF}},{\rm{d}}}} \right) \in {\cal R}_{{\rm{ZF}}}^{\rm{d}}$, which satisfies $\left( {\hat r_1^B,\hat r_2^B} \right) \preceq \left( {\bar r_1^{{\rm{ZF,d}}},\bar r_2^{{\rm{ZF,d}}}} \right)$. Let $\left( {\bar r_1^{{\rm{ZF}},s},\bar r_2^{{\rm{ZF}},s}} \right) = \left( {\hat r_1^B,r_2^B} \right)$. Then, it can be obtained that $\left( {\bar r_1^{{\rm{ZF,d}}},\bar r_2^{{\rm{ZF,d}}}} \right) \succ \left( {\bar r_1^{{\rm{ZF}},s},\bar r_2^{{\rm{ZF}},s}} \right)$, which thus completes the proof.
\end{proof}

Proposition 3 reveals the sufficient condition when dynamic beamforming can enlarge the achievable rate region under the ZF-based transmission scheme. Specifically, as long as the feasible IRS phase-shifts cannot achieve fully orthogonal channels for the two users, employing dynamic beamforming is always capable of improving the achievable rate. However, adjusting the IRS phase-shifts and the active beamforming over time also incurs more signalling overhead as well as higher complexity of hardware implementation. For the given IRS phase-shift matrix ${\bf{\Theta }}$, the achievable rate region of the IRS aided ZF scheme with static bemforming configuration is denoted by ${{\cal R}_{{\rm{ZF}}}}\left( {\bf{\Theta }} \right)$. To further shed light on the impact of the IRS on the gain attained by dynamic beamforming, we have the following proposition under the two-user case.

\begin{pos}
${{\cal R}_{{\rm{ZF}}}}\left( {\bf{\Theta }} \right)$ is convex if and only if
\begin{align}\label{ro_two_user}
\rho \left( {\bf{\Theta }} \right) \buildrel \Delta \over = \frac{{\left| {\left( {{\bf{h}}_{d,1}^H + {\bf{h}}_{r,1}^H{\bf{\Theta G}}} \right){{\left( {{\bf{h}}_{d,1}^H + {\bf{h}}_{r,1}^H{\bf{\Theta G}}} \right)}^H}} \right|}}{{\left\| {{\bf{h}}_{d,1}^H + {\bf{h}}_{r,1}^H{\bf{\Theta G}}} \right\|\left\| {{\bf{h}}_{d,2}^H + {\bf{h}}_{r,2}^H{\bf{\Theta G}}} \right\|}} = 0.
\end{align}
\end{pos}

\begin{figure*}[t!]
\centering
\subfigure[$\rho \left( {\bf{\Theta }} \right) > 0$.]{\label{org_a}
\includegraphics[width= 2.6in, height=2in]{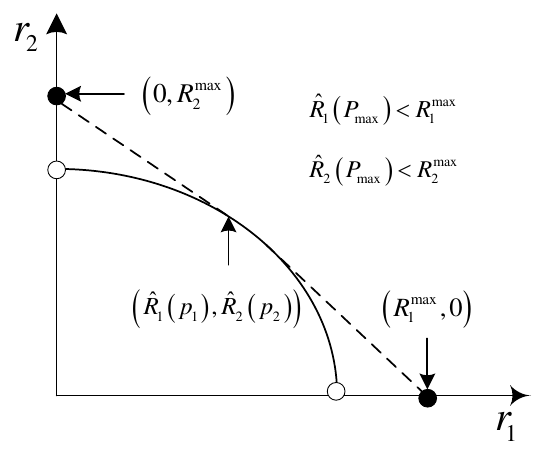}}
\subfigure[$\rho \left( {\bf{\Theta }} \right) = 0$.]{\label{org_b}
\includegraphics[width= 2.6in, height=2in]{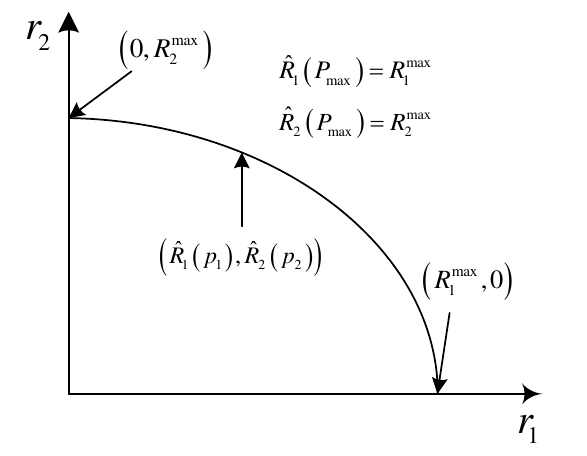}}
\setlength{\abovecaptionskip}{0.4cm}
\caption{{Illustration of Proposition 4.}}\label{fig_cor}
\vspace{-20pt}
\end{figure*}

\begin{proof}
Under the given ${\bf{\Theta }}$, ${{\cal R}_{{\rm{ZF}}}}\left( {\bf{\Theta }} \right)$ denotes the set of all rate-pairs $\left( {{r_1},{r_2}} \right)$ achieved by the ZF-based transmission scheme. Specifically, ${{\cal R}_{{\rm{ZF}}}}\left( {\bf{\Theta }} \right)$ can be expressed as
\begin{align}\label{full_set}
{{\cal R}_{{\rm{ZF}}}}\left( {\bf{\Theta }} \right) = {{\tilde {\cal R}}_{{\rm{ZF}}}}\left( {\bf{\Theta }} \right)\bigcup {\left( {0,R_2^{\max }} \right)} \bigcup {\left( {R_1^{\max },0} \right)},
\end{align}
where $R_k^{\max } = {\log _2}\left( {1 + {P_{\max }}{{\left\| {{{\bf{h}}_k}\left( {\bf{\Theta }} \right)} \right\|}^2}/{\sigma ^2}} \right),k = 1,2$ and ${{\tilde {\cal R}}_{{\rm{ZF}}}}\left( {\bf{\Theta }} \right)$ is characterized by the following inequities:
\begin{align}\label{inequiality_cha}
&0 \le {r_1} \le {{\hat R}_1}\left( {{p_1}} \right) \buildrel \Delta \over = {\log _2}\left( {1 + \frac{{{p_1}{{\left\| {{{\bf{h}}_1}\left( {\bf{\Theta }} \right)} \right\|}^2}}}{{{\sigma ^2}}}\left( {1 - {\rho ^2}\left( {\bf{\Theta }} \right)} \right)} \right), \\
&0 \le {r_2} \le {{\hat R}_2}\left( {{p_2}} \right) \buildrel \Delta \over = {\log _2}\left( {1 + \frac{{{p_2}{{\left\| {{{\bf{h}}_2}\left( {\bf{\Theta }} \right)} \right\|}^2}}}{{{\sigma ^2}}}\left( {1 - {\rho ^2}\left( {\bf{\Theta }} \right)} \right)} \right), \\
&{p_1} + {p_2} \le {P_{\max }}.
\end{align}
Note that all the points, i.e., $\left( {{{\hat R}_1}\left( {{p_1}} \right),{{\hat R}_2}\left( {{p_2}} \right)} \right)$, on the Pareto boundary can be obtained by varying ${p_1}$ or ${p_2}$ under the constraint that ${p_1} + {p_2} = {P_{\max }}$. Hence, ${{{\hat R}_2}\left( {{p_2}} \right)}$ is a decreasing function of ${{{\hat R}_1}\left( {{p_1}} \right)}$. By taking the second order derivative of ${{{\hat R}_2}}$ with respect to ${{{\hat R}_1}}$, we have
\begin{align}\label{second_order_dev}
\frac{{{\partial ^2}{{\hat R}_2}\left( {{p_2}} \right)}}{{{\partial ^2}{{\hat R}_1}\left( {{p_1}} \right)}} &= \left( {\partial \left( {\frac{{\partial {{\hat R}_2}\left( {{p_2}} \right)}}{{\partial {{\hat R}_1}\left( {{p_1}} \right)}}} \right)/\partial {p_1}} \right)\left( {\partial {p_1}/\partial {{\hat R}_1}\left( {{p_1}} \right)} \right) \nonumber\\
& =  - \frac{{{\Gamma _2}\left( {\bf{\Theta }} \right)\left( {{\Gamma _1}\left( {\bf{\Theta }} \right) + {\Gamma _2}\left( {\bf{\Theta }} \right) + {P_{\max }}{\Gamma _1}\left( {\bf{\Theta }} \right){\Gamma _2}\left( {\bf{\Theta }} \right)} \right)}}{{{{\left( {\left( {1 + {p_2}{\Gamma _2}\left( {\bf{\Theta }} \right)} \right){\Gamma _1}\left( {\bf{\Theta }} \right)} \right)}^2}}}\left( {1 + {p_1}{\Gamma _1}\left( {\bf{\Theta }} \right)} \right) < 0,
\end{align}
where ${\Gamma _k}\left( {\bf{\Theta }} \right) = {\left\| {{{\bf{h}}_k}\left( {\bf{\Theta }} \right)} \right\|^2}\left( {1 - {\rho ^2}\left( {\bf{\Theta }} \right)} \right)/{\sigma ^2},k = 1,2$. Based on \eqref{second_order_dev}, ${{{\hat R}_2}}$ is a concave function with respect to ${{{\hat R}_1}}$. Thus, ${{\tilde {\cal R}}_{{\rm{ZF}}}}\left( {\bf{\Theta }} \right)$ is convex. It can be easily verified that ${{\cal R}_{{\rm{ZF}}}}\left( {\bf{\Theta }} \right) = {{\tilde {\cal R}}_{{\rm{ZF}}}}\left( {\bf{\Theta }} \right)$ under the condition that \eqref{ro_two_user} is satisfied. As such, it follows that ${{\cal R}_{{\rm{ZF}}}}\left( {\bf{\Theta }} \right)$ is a convex set if \eqref{ro_two_user} holds. On the other hand, if \eqref{ro_two_user} is not satisfied, ${\log _2}\left( {1 + {P_{\max }}{{\left\| {{{\bf{h}}_k}\left( {\bf{\Theta }} \right)} \right\|}^2}/{\sigma ^2}} \right) > {{\hat R}_k}\left( {{p_k}} \right),k = 1,2$ always holds, which leads to
\begin{align}\label{set_relation}
\left( {\left( {R_1^{\max },0} \right)\bigcup {\left( {0,R_2^{\max }} \right)} } \right) \notin {{\tilde {\cal R}}_{{\rm{ZF}}}}\left( {\bf{\Theta }} \right).
\end{align}
Based on \eqref{full_set} and \eqref{set_relation}, ${{\cal R}_{{\rm{ZF}}}}\left( {\bf{\Theta }} \right)$ is non-convex if $\rho \left( {\bf{\Theta }} \right) > 0$. Thus, we complete the proof.
\end{proof}

As illustrated in Fig. \ref{org_a}, the achievable rate region achieved by employing static beamforming is a non-convex set if the channels of users are not orthogonal. Thus, the convex hull operation is needed to further enlarge its achievable rate region. Specifically, the rate pairs on the dash line in Fig. \ref{org_a} can be obtained by using dynamic beamforming. In contrast, if the channels of users are forced to be orthogonal by the IRS, the corresponding rate region is convex and thus the convex hull operation is not needed, which suggests that dynamic beamforming cannot provide performance gains. Moreover, the gain attained by dynamic beamforming becomes marginal as $\Delta {R_k} \buildrel \Delta \over = R_k^{\max } - {{\hat R}_k}\left( {{P_{\max }}} \right)$ decreases. It is worth noting that $\Delta {R_k}$ can be significantly reduced by employing the IRS to degrade $\rho \left( {\bf{\Theta }} \right)$. Therefore, the demand for implementing dynamic beamforming can be effectively reduced with the help of employing the IRS to minimize the channel correlation, which is helpful for lowering the system signaling overhead.

\subsection{Performance Gap between IRS aided ZF and DPC?}
In addition to reducing the channel correlation, IRS is also capable of providing high passive beamforming gain to further strengthen users' effective channels. This subsection aims to investigate the relative performance gain of IRS aided DPC over ZF as $N \to \infty$ from the perspective of high passive beamforming gain. Under the case of $\bar N = 1$, we have the following theorem.
\begin{thm}
Suppose that ${\bf{h}}_{r,k}^H \sim {\cal C}{\cal N}\left( {{\bf{0}},\rho _{r,k}^{^2}{{\bf{I}}_N}} \right)$ and ${\bf{G}} \sim {\cal C}{\cal N}\left( {{\bf{0}},\rho _g^{^2}{{\bf{I}}_{MN}}} \right)$. Under randomly given IRS phase-shifts, i.e., $\left\{ {{\theta _n}} \right\}$, we have
\begin{align}\label{relation}
\eta  = \mathop {\lim }\limits_{N \to \infty } \frac{{{{\mathbb{E}}}\left[ {{R_{{\rm{DPC}}}}} \right] - {{\mathbb{E}}}\left[ {{R_{ZF}}} \right]}}{{{{\mathbb{E}}}\left[ {{R_{{\rm{ZF}}}}} \right]}} = 0,
\end{align}
where ${{R_{{\rm{DPC}}}}}$ and ${{R_{{\rm{ZF}}}}}$ are the sum-rates of DPC and ZF schemes, respectively.
\end{thm}
\begin{proof}
The equivalent channel for each user can be written as
\begin{align}\label{equivalent_channel}
{\bf{h}}_k^H = {\bf{h}}_{r,k}^H{\bf{\Theta G}} = \sum\nolimits_{n = 1}^N {h_{r,k}^H\left( n \right)} {{\bf{g}}^H}\left( n \right){e^{j{\theta _n}}},
\end{align}
where ${h_{r,k}^H\left( n \right)}$ is the $n$-th element of ${\bf{h}}_{r,k}^H$ and ${{\bf{g}}^H}\left( n \right)$ is the $n$-th row of ${\bf{G}}$. Accordingly, the $m$-th element of ${\bf{h}}_k^H$ is given by
\begin{align}\label{equivalent_channel}
h_k^H\left( m \right) = \sum\nolimits_{n = 1}^N {h_{r,k}^H\left( n \right)} G\left( {n,m} \right){e^{j{\theta _n}}},
\end{align}
where $G\left( {n,m} \right)$ denotes the $m$-th element of ${{\bf{g}}^H}\left( n \right)$. As $N \to \infty$, we have
\begin{align}\label{equivalent_channe2}
h_k^H\left( m \right)/\left( {\rho _{r,k}^2\rho _g^2N} \right) \sim {\cal C}{\cal N}\left( {0,1} \right),
\end{align}
based on the central limit theorem and ${\mathop{\rm Var}\nolimits} \left( {h_{r,k}^H\left( n \right){{\bf{g}}^H}\left( n \right){e^{j{\theta _n}}}} \right) = \rho _{r,k}^{^2}\rho _g^{^2}$. As such, we have $h_k^H\left( m \right) \sim {\cal C}{\cal N}\left( {0,\rho _{r,k}^{^2}\rho _g^{^2}N} \right)$. Then, we prove \eqref{relation} by showing that $\eta  \ge 0$ and $\eta  \le 0$. The proof of $\eta  \ge 0$ is intuitive since ${R_{{\rm{DPC}}}} \ge {R_{{\rm{ZF}}}}$ naturally holds under any given channel conditions. Then, we focus on showing that $\eta  \le 0$ as follows. The transmit power for each user is denoted by $\left\{ {{p_k}} \right\}$ and then ${\mathop{{\mathbb{E}}}\nolimits} \left[ {{R_{{\rm{DPC}}}}} \right]$ is upper bounded by
\begin{align}\label{upper_bound}
{\mathop{{\mathbb{E}}}\nolimits} \left[ {{R_{{\rm{DPC}}}}} \right] &\le {\mathop{{\mathbb{E}}}\nolimits} \left[ {\sum\nolimits_{k = 1}^K {{{\log }_2}\left( {1 + \frac{{{p_k}{{\left\| {{{\bf{h}}_k}} \right\|}^2}}}{{{\sigma ^2}}}} \right)} } \right]\mathop  \le \limits^{\left( a \right)} \sum\nolimits_{k = 1}^K {{\mathbb{E}}\left[ {{{\log }_2}\left( {1 + \frac{{{P_{\max }}\rho _{r,k}^{^2}\rho _g^{^2}{{\left\| {{{{\bf{\tilde h}}}_k}} \right\|}^2}}}{{{\sigma ^2}}}} \right)} \right]} \nonumber\\
&\mathop  \le \limits^{\left( b \right)} \sum\nolimits_{k = 1}^K {{{\log }_2}\left( {1 \!\!+\!\! \frac{{{P_{\max }}\rho _{r,k}^{^2}\rho _g^{^2}{\mathop{{{\mathbb{E}}}}\nolimits} \left\{ {{{\left\| {{{{\bf{\tilde h}}}_k}} \right\|}^2}} \right\}}}{{{\sigma ^2}}}} \right)}  \le K{\log _2}\left( {1 \!+\! \frac{{{P_{\max }}\mathord{\buildrel{\lower3pt\hbox{$\scriptscriptstyle\frown$}}
\over \rho } _r^{^2}\rho _g^{^2}{\mathop{{\mathbb{E}}}\nolimits} \left\{ {{{\left\| {{{{\bf{\tilde h}}}_k}} \right\|}^2}} \right\}}}{{{\sigma ^2}}}} \right)\nonumber\\
&\buildrel \Delta \over = {{\mathord{\buildrel{\lower3pt\hbox{$\scriptscriptstyle\frown$}}
\over R} }_{{\rm{DPC}}}},
\end{align}
where ${{{\bf{\tilde h}}}_k} = {{{{\bf{h}}_k}} \mathord{\left/
 {\vphantom {{{{\bf{h}}_k}} {\left( {\rho _{r,k}^{^2}\rho _g^{^2}} \right)}}} \right.
 \kern-\nulldelimiterspace} {\left( {\rho _{r,k}^{^2}\rho _g^{^2}} \right)}}$, $\mathord{\buildrel{\lower3pt\hbox{$\scriptscriptstyle\frown$}}
\over \rho } _r^2 = \max \left\{ {\rho _{r,1}^{^2}, \ldots ,\rho _{r,K}^{^2}} \right\}$, (a) follows from ${p_k} \le {P_{\max }}$, and (b) follows from Jensen's inequality based on the fact that $g\left( x \right) = {\log _2}\left( {1 + x} \right)$ is a concave function with respect to $x$. Since $h_k^H\left( m \right) \sim {\cal C}{\cal N}\left( {0,\rho _{r,k}^{^2}\rho _g^{^2}N} \right)$, it can be readily verified that ${\left\| {{{{\bf{\tilde h}}}_k}} \right\|^2} \sim {\rm{Gamma}}\left( {M,N} \right)$. Then, the upper bound of ${\mathop{{\mathbb{E}}}\nolimits} \left[ {{R_{{\rm{DPC}}}}} \right]$ can be obtained as
\begin{align}\label{upper_bound1}
{{\mathord{\buildrel{\lower3pt\hbox{$\scriptscriptstyle\frown$}}
\over R} }_{{\rm{DPC}}}} = K{\log _2}\left( {1 + \frac{{{P_{\max }}\mathord{\buildrel{\lower3pt\hbox{$\scriptscriptstyle\frown$}}
\over \rho } _r^{^2}\rho _g^{^2}NM}}{{{\sigma ^2}}}} \right).
\end{align}
For the ZF scheme, we aim to find the lower bound of ${{\mathop{{\mathbb{E}}}\nolimits}  \left[ {{R_{{\rm{ZF}}}}} \right]}$. The expression of ${R_{{\rm{ZF}}}}$ can be expressed as
\begin{align}\label{ZF_rate}
{R_{{\rm{ZF}}}} = \sum\nolimits_{k = 1}^K {{{\log }_2}\left( {1 + \frac{{{p_k}{{\left| {{\bf{h}}_k^H{{\bf{w}}_k}} \right|}^2}}}{{{\sigma ^2}}}} \right)}  = \sum\nolimits_{k = 1}^K {{{\log }_2}\left( {1 + \frac{{{p_k}\rho _{r,k}^{^2}\rho _g^{^2}{{\left| {{\bf{\tilde h}}_k^H{{\bf{w}}_k}} \right|}^2}}}{{{\sigma ^2}}}} \right)},
\end{align}
where ${{{\bf{w}}_k}}$ equals the $k$-th normalized column of ${\bf{\tilde H}}{\left( {{{{\bf{\tilde H}}}^H}{\bf{\tilde H}}} \right)^{ - 1}}$ and ${\bf{\tilde H}} = \left[ {{{{\bf{\tilde h}}}_1}, \ldots ,{{{\bf{\tilde h}}}_K}} \right]$. As $N \to \infty $, ${{{\bf{\tilde h}}}_k} \sim {\cal C}{\cal N}\left( {{\bf{0}},N{{\bf{I}}_M}} \right)$. Therefore, ${{{\bf{\tilde h}}}_k}$'s are linearly independent with each other with probability one. ${{{\left| {{\bf{\tilde h}}_k^H{{\bf{w}}_k}} \right|}^2}}$ is the power of an isotropic dimensional beamforming space. Then, we have ${\left| {{\bf{\tilde h}}_k^H{{\bf{w}}_k}} \right|^2} \sim {\mathop{\rm Gamma}\nolimits} \left( {M - K + 1,N} \right)$ as $N \to \infty $. The lower bound of ${{\mathop{{\mathbb{E}}}\nolimits} \left[ {{R_{{\rm{ZF}}}}} \right]}$ can be derived as
\begin{align}\label{ZF_bound}
{\mathop{{\mathbb{E}}}\nolimits} \left[ {{R_{{\rm{ZF}}}}} \right] &= \sum\nolimits_{k = 1}^K {{\mathop{{\mathbb{E}}}\nolimits} \left[ {{{\log }_2}\left( {1 + \frac{{{p_k}\rho _{r,k}^{^2}\rho _g^{^2}}}{{{\sigma ^2}}}\frac{1}{{{{\left| {{\bf{\tilde h}}_k^H{{\bf{w}}_k}} \right|}^2}}}} \right)} \right]}\nonumber\\
&\mathop  \ge \limits^{\left( a \right)} \sum\nolimits_{k = 1}^K {{{\log }_2}\left( {1 + \frac{{{p_k}\rho _{r,k}^{^2}\rho _g^{^2}}}{{{\sigma ^2}}}\frac{1}{{{\mathop{{\mathbb{E}}}\nolimits} \left\{ {{{\left| {{\bf{\tilde h}}_k^H{{\bf{w}}_k}} \right|}^2}} \right\}}}} \right)}\nonumber\\
&\mathop  \ge \limits^{\left( b \right)} K{\log _2}\left( {1 + \frac{{{P_{\max }}\mathord{\buildrel{\lower3pt\hbox{$\scriptscriptstyle\smile$}}
\over \rho } _r^2\rho _g^{^2}N\left( {M - K} \right)}}{{K{\sigma ^2}}}} \right) \buildrel \Delta \over = {{\mathord{\buildrel{\lower3pt\hbox{$\scriptscriptstyle\smile$}}
\over R} }_{{\rm{ZF}}}},
\end{align}
where (a) follows from Jensen's inequality based on the fact that $h\left( x \right) = {\log _2}\left( {1 + {1 \mathord{\left/
 {\vphantom {1 x}} \right.
 \kern-\nulldelimiterspace} x}} \right)$ is a convex function with respect to $x$, and (b) follows from that equal power allocation is suboptimal for maximizing the sum-rate and $\mathord{\buildrel{\lower3pt\hbox{$\scriptscriptstyle\smile$}}
\over \rho } _r^{^2} = \min \left\{ {\rho _{r,1}^{^2}, \ldots ,\rho _{r,K}^{^2}} \right\}$. As $N \to \infty$, we have
\begin{align}\label{equivalent}
&{{\mathord{\buildrel{\lower3pt\hbox{$\scriptscriptstyle\frown$}}
\over R} }_{{\rm{DPC}}}} \sim K{\log _2}\left( {\frac{{{P_{\max }}\mathord{\buildrel{\lower3pt\hbox{$\scriptscriptstyle\frown$}}
\over \rho } _r^{^2}\rho _g^{^2}NM}}{{{\sigma ^2}}}} \right) = K{\log _2}\left( N \right) + K{\log _2}\left( {\frac{{{P_{\max }}\mathord{\buildrel{\lower3pt\hbox{$\scriptscriptstyle\frown$}}
\over \rho } _r^2\rho _g^2M}}{{{\sigma ^2}}}} \right)\nonumber\\
&{{\mathord{\buildrel{\lower3pt\hbox{$\scriptscriptstyle\smile$}}
\over R} }_{{\rm{ZF}}}} \sim K{\log _2}\left( {\frac{{{P_{\max }}\mathord{\buildrel{\lower3pt\hbox{$\scriptscriptstyle\smile$}}
\over \rho } _r^2\rho _g^{^2}N\left( {M - K} \right)}}{{K{\sigma ^2}}}} \right) = K{\log _2}\left( N \right) + K{\log _2}\left( {\frac{{{P_{\max }}\mathord{\buildrel{\lower3pt\hbox{$\scriptscriptstyle\frown$}}
\over \rho } _r^2\rho _g^2\left( {M - K} \right)}}{{{\sigma ^2}}}} \right),
\end{align}
which leads to
\begin{align}\label{limit}
\mathop {\lim }\limits_{N \to \infty } \frac{{{{\mathbb{E}}}\left[ {{R_{{\rm{DPC}}}}} \right]}}{{{{\mathbb{E}}}\left[ {{R_{{\rm{ZF}}}}} \right]}} \le \mathop {\lim }\limits_{N \to \infty } \frac{{{{\mathord{\buildrel{\lower3pt\hbox{$\scriptscriptstyle\frown$}}
\over R} }_{{\rm{DPC}}}}}}{{{{\mathord{\buildrel{\lower3pt\hbox{$\scriptscriptstyle\smile$}}
\over R} }_{{\rm{ZF}}}}}} = \mathop {\lim }\limits_{N \to \infty } \frac{{K{{\log }_2}\left( N \right) + K{{\log }_2}\left( {\frac{{{P_{\max }}\mathord{\buildrel{\lower3pt\hbox{$\scriptscriptstyle\frown$}}
\over \rho } _r^2\rho _g^2M}}{{{\sigma ^2}}}} \right)}}{{K{{\log }_2}\left( N \right) + K{{\log }_2}\left( {\frac{{{P_{\max }}\mathord{\buildrel{\lower3pt\hbox{$\scriptscriptstyle\frown$}}
\over \rho } _r^2\rho _g^2\left( {M - K} \right)}}{{{\sigma ^2}}}} \right)}} = 1.
\end{align}

According to \eqref{limit}, we have $\eta  \le 0$. Based on the facts that $\eta  \ge 0$ and $\eta  \le 0$, we have $\eta  = 0$, which thus completes the proof.
\end{proof}

Theorem 1 reveals that the relative sum-rate gain achieved by IRS aided DPC over ZF becomes vanished as the number of IRS elements increases. This is because despite suffering the loss of active beamforming gain at the BS, the sum multiplexing gain of the ZF scheme from the DoF perspective is ${\rm{Do}}{{\rm{F}}_{{\rm{ZF}}}} = \mathop {\lim }\limits_{{P_{\max }} \to \infty } \frac{{{R_{{\rm{ZF}}}}}}{{{{\log }_2}{P_{\max }}}} = K$, which is the same as DPC \cite{1424315}. This indicates that the sum rate of ZF can approach that of DPC in the high SNR regime. In the IRS aided wireless system, the high passive beamforming gain introduced by the IRS dominates the the loss of active beamforming gain for ZF, which artificially generates the equivalent high SNR regime even under the moderate value of ${{P_{\max }}}$.

\section{Numerical results}
In this section, numerical results are provided to validate our theoretical findings and to draw useful insights. The BS and the IRS are located at (0, 0, 0) meter (m) and (50, 3, 0) m, respectively. The users are located on the half-circle centered of the IRS with a radius of 5 m. For both the BS-IRS and IRS-user links, the path-loss exponents are set to $2.2$, while those of the BS-user direct links are set to $3.4$. Besides, we set $M = 4$, $\bar N = 4$, $b = 2$, ${P_{\max }} = 20~{\rm{dBm}}$, and ${\sigma ^2} =  - 80~{\rm{dBm}}$. Rayleigh fading is considered to characterize the small-scale fading for all links. Furthermore, we set the path loss at the reference distance of 1 m as 30 dB.

\subsection{Capacity and Achievable Rate Region}
To show the capacity and rate regions achieved by DPC and ZF, we consider $K=2$ in this subsection. Specifically, we consider the following schemes for comparison: 1) \textbf{DPC, exact bound}: Algorithm 1 is employed to obtain the Pareto boundary of the capacity region achieved by DPC, which serves as the performance upper-bound; 2) \textbf{DPC, inner bound}: Algorithm 2 is applied to obtain the inner rate boundary of the capacity region achieved by DPC; 3) \textbf{ZF, static BF}: Algorithm 1 is extended to obtain the Pareto boundary of the rate region achieved by ZF under the static beamforming configuration; 4) \textbf{ZF, dynamic BF}: the Pareto boundary of the rate region achieved by ZF under the dynamic beamforming configuration; 5) \textbf{TDMA}: orthogonal TSs are allocated to two users and the dedicated BS/IRS beamforming is designed to maximize the desired signal power of each user. We define ${\rho _d} \buildrel \Delta \over = \left| {{\bf{h}}_{d,1}^H{{\bf{h}}_{d,2}}} \right|/\left( {\left\| {{{\bf{h}}_{d,1}}} \right\|\left\| {{{\bf{h}}_{d,2}}} \right\|} \right)$ as the channel correlation coefficient of the direct links. Then, the capacity/achivable rate regions of the aforementioned schemes will be evaluated under different values of $\rho _d^2$ to shed light on how the IRS collaborate with the BS to improve the capacity performance.

\subsubsection{High channel correlation} In Fig. \ref{region1}, we present the capacity and achievable rate regions of the considered schemes under the case of $\rho _d^2 = 0.8$, which represents one typical channel condition of high correlation. The corresponding square of channel correlation coefficient after integrating IRS, i.e., ${\rho ^2}\left( {\bf{\Theta }} \right)$, is also depicted in Fig. \ref{correlation1}. In the absence of the IRS, it is observed from Fig. \ref{region1} that ZF under the static beamforming configuration suffers substantial performance loss compared to the capacity-achieving DPC scheme and even performs worse than the TDMA scheme. This is mainly because the active beamforming gain is weak due to the high channel correlation. In this case, employing dynamic beamforming configuration is reduced to the TDMA scheme and can significantly improve the achievable rate region. After integrating the IRS, the capacity/achievable rate regions of all schemes are effectively enlarged. Moreover, it is noticed that the performance gains achieved by dynamic beamforming become marginal since the channel correlation coefficient is significantly reduced as indicated by Fig. \ref{correlation1}, which is consistent with our analysis in Proposition 3 and Proposition 4. This indicates that the demand for employing dynamic beamforming configuration can be effectively alleviated by deploying the IRS, which is helpful for reducing the signalling/operating overhead in practical systems. Additionally,  dynamic beamforming is generally not needed for DPC since the observed capacity regions achieved by DPC are nearly convex-shape.

\begin{figure*}[t!]
\centering
\subfigure[Capacity/rate region.]{\label{region1}
\includegraphics[width= 2.5in, height=2in]{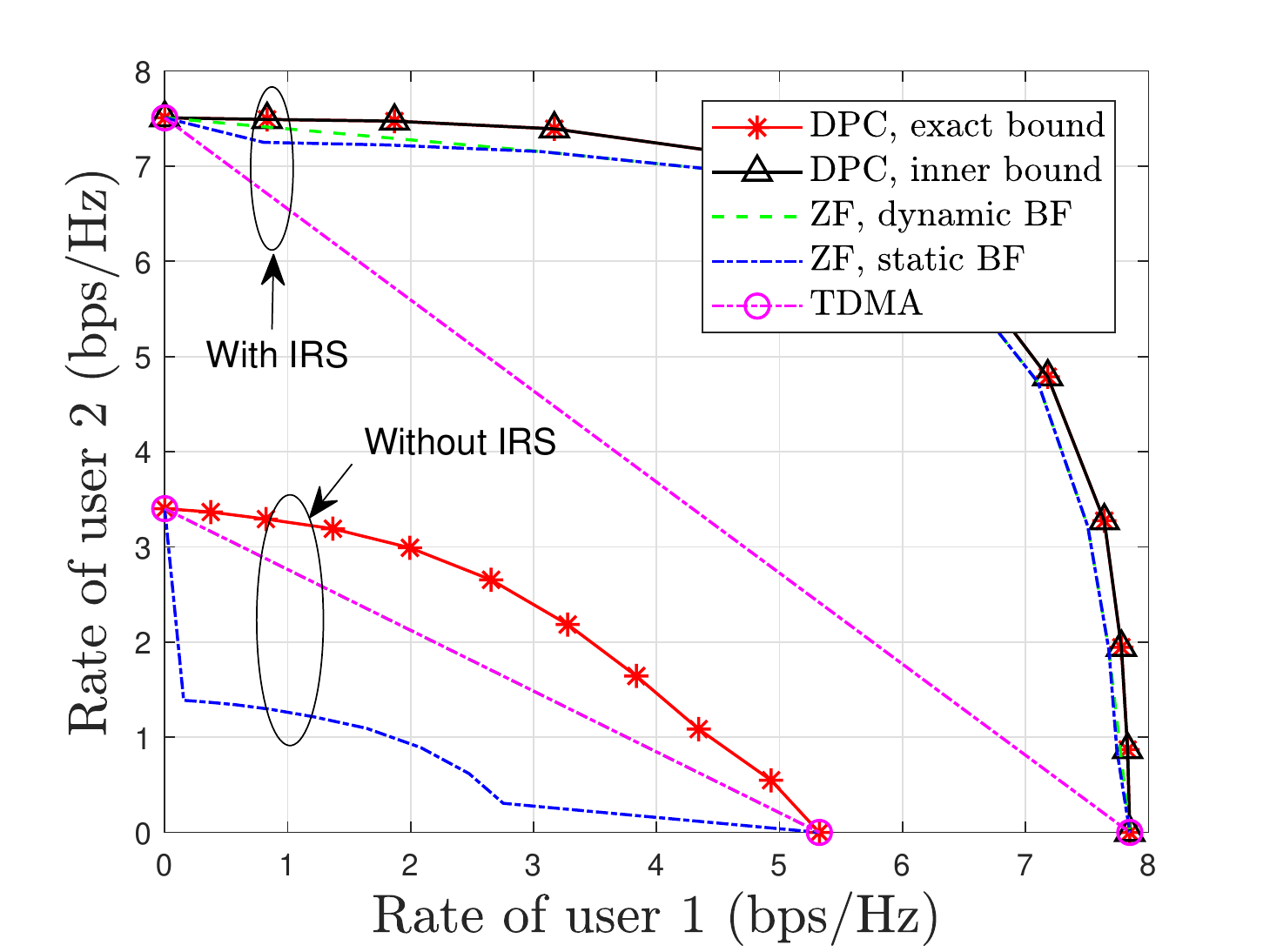}}
\subfigure[Resulting channel correlation.]{\label{correlation1}
\includegraphics[width= 2.5in, height=2in]{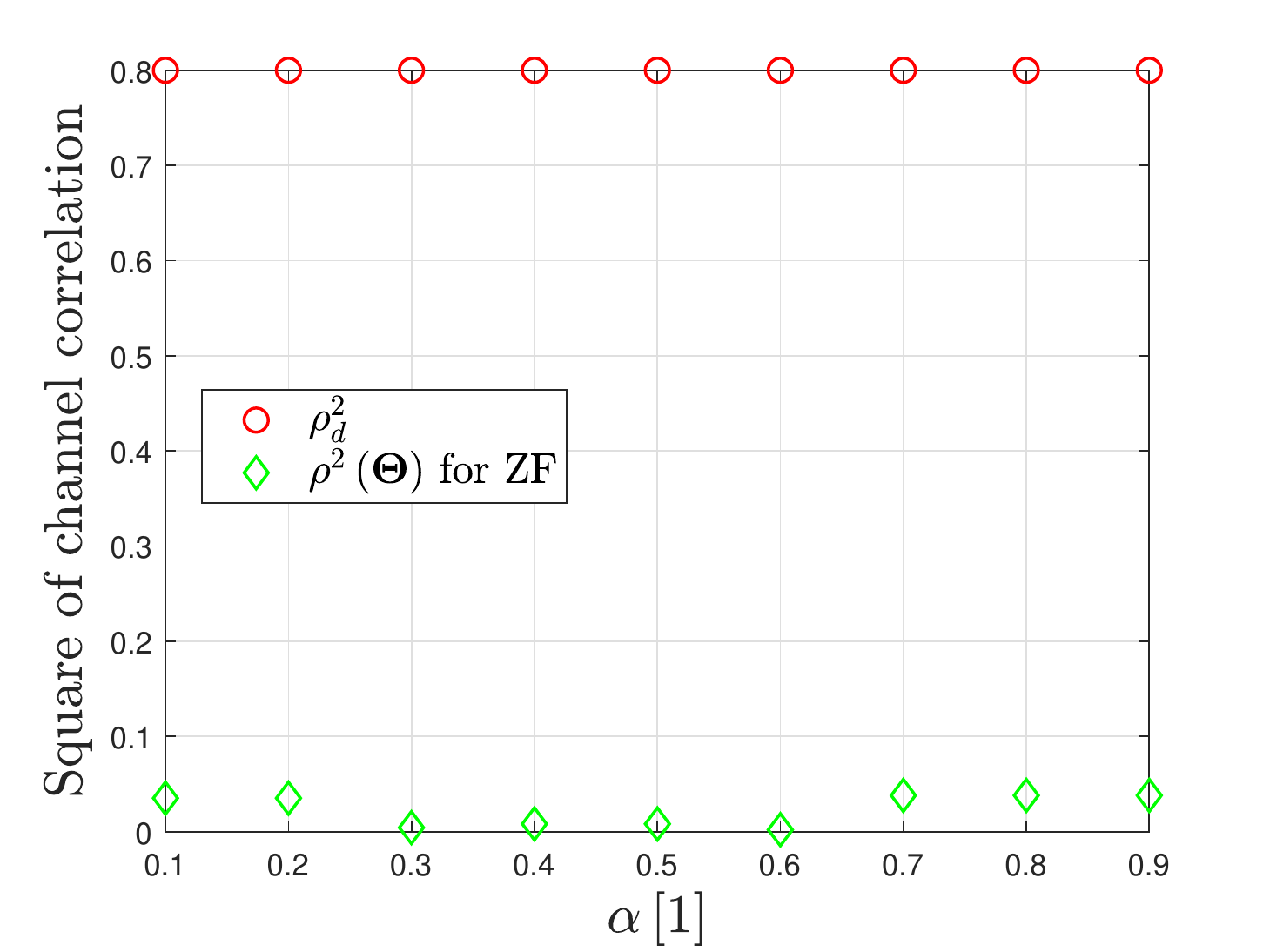}}
\setlength{\abovecaptionskip}{0.4cm}
\caption{{Capacity/rate region for IRS aided MISO broadcast channel under the case of $\rho _d^2 = 0.8$.}}\label{rate_region_high}
\vspace{-16pt}
\end{figure*}

Thanks to the lower channel correlation introduced by properly designing the IRS phase-shifts, the rate achieved by ZF is able to approach the exact bound achieved by DPC. It suggests that the sophisticated transmission scheme, i.e., DPC, can be potentially replaced by the easy-implementation scheme, i.e., ZF, , which is essential for simplifying the transceiver design. Finally, it is observed that the inner bound is close to the exact bound, which verifies the effectiveness of the proposed Algorithm 2 for the inner bound characterization.

\begin{figure*}[t!]
\centering
\subfigure[Capacity/rate region.]{\label{region22}
\includegraphics[width= 2.5in, height=2in]{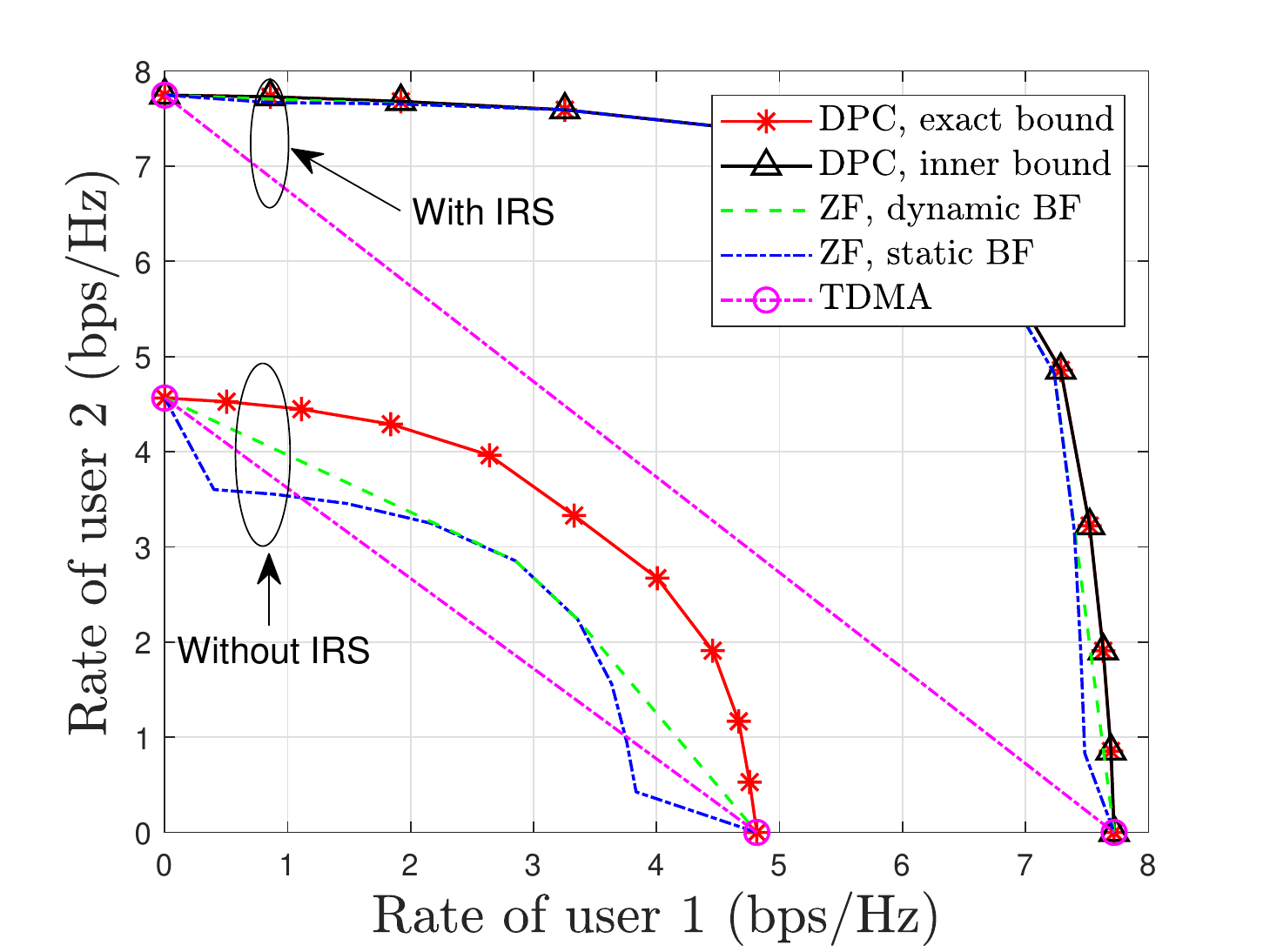}}
\subfigure[Resulting channel correlation.]{\label{correlation2}
\includegraphics[width= 2.5in, height=2in]{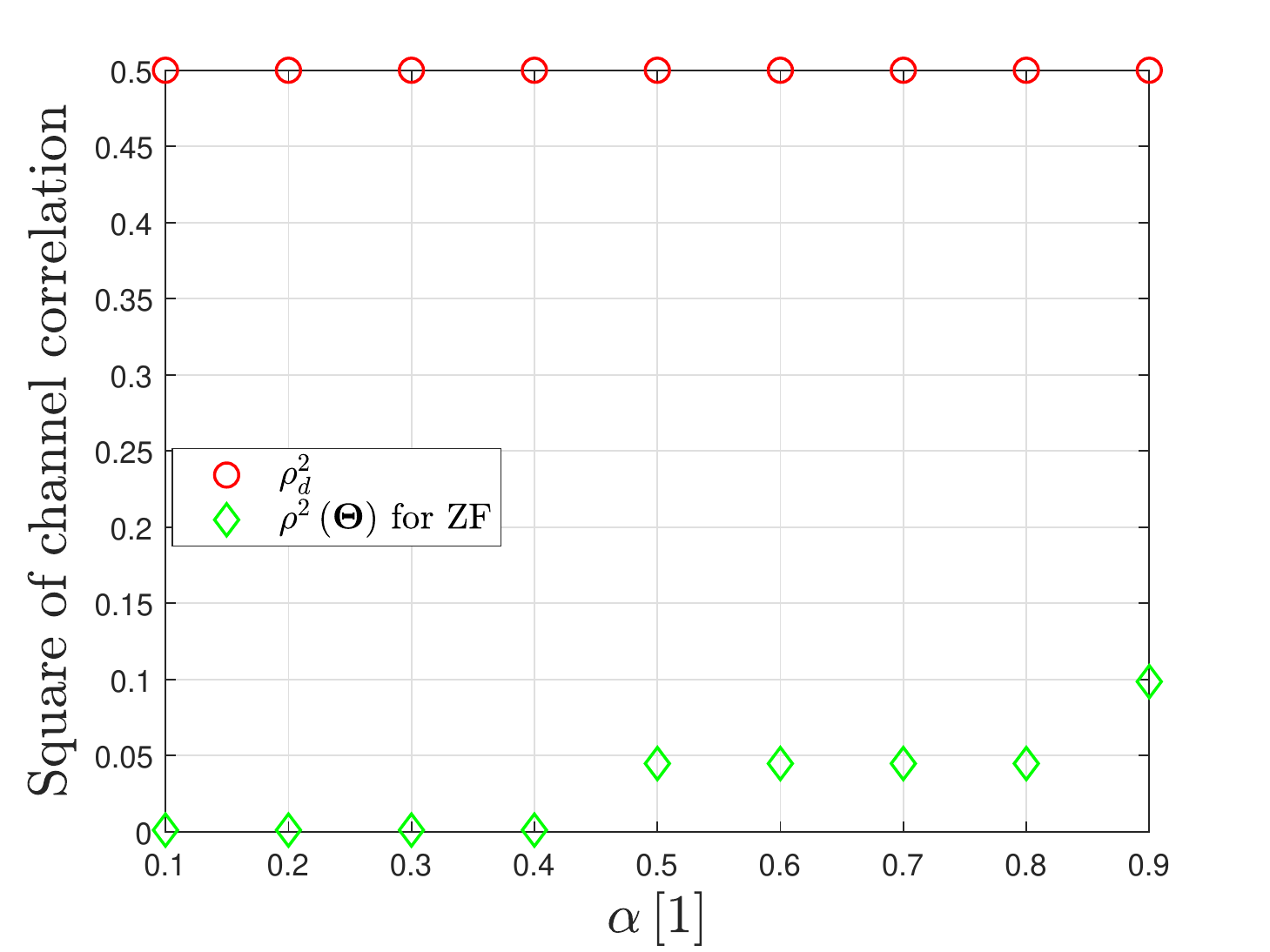}}
\setlength{\abovecaptionskip}{0.4cm}
\caption{{Capacity/rate region for IRS aided MISO broadcast channel under the case of $\rho _d^2 = 0.5$.}}\label{rate_region_moderate}
\vspace{-16pt}
\end{figure*}

\begin{figure*}[t!]
\centering
\subfigure[Capacity/rate region.]{\label{region3}
\includegraphics[width= 2.5in, height=2in]{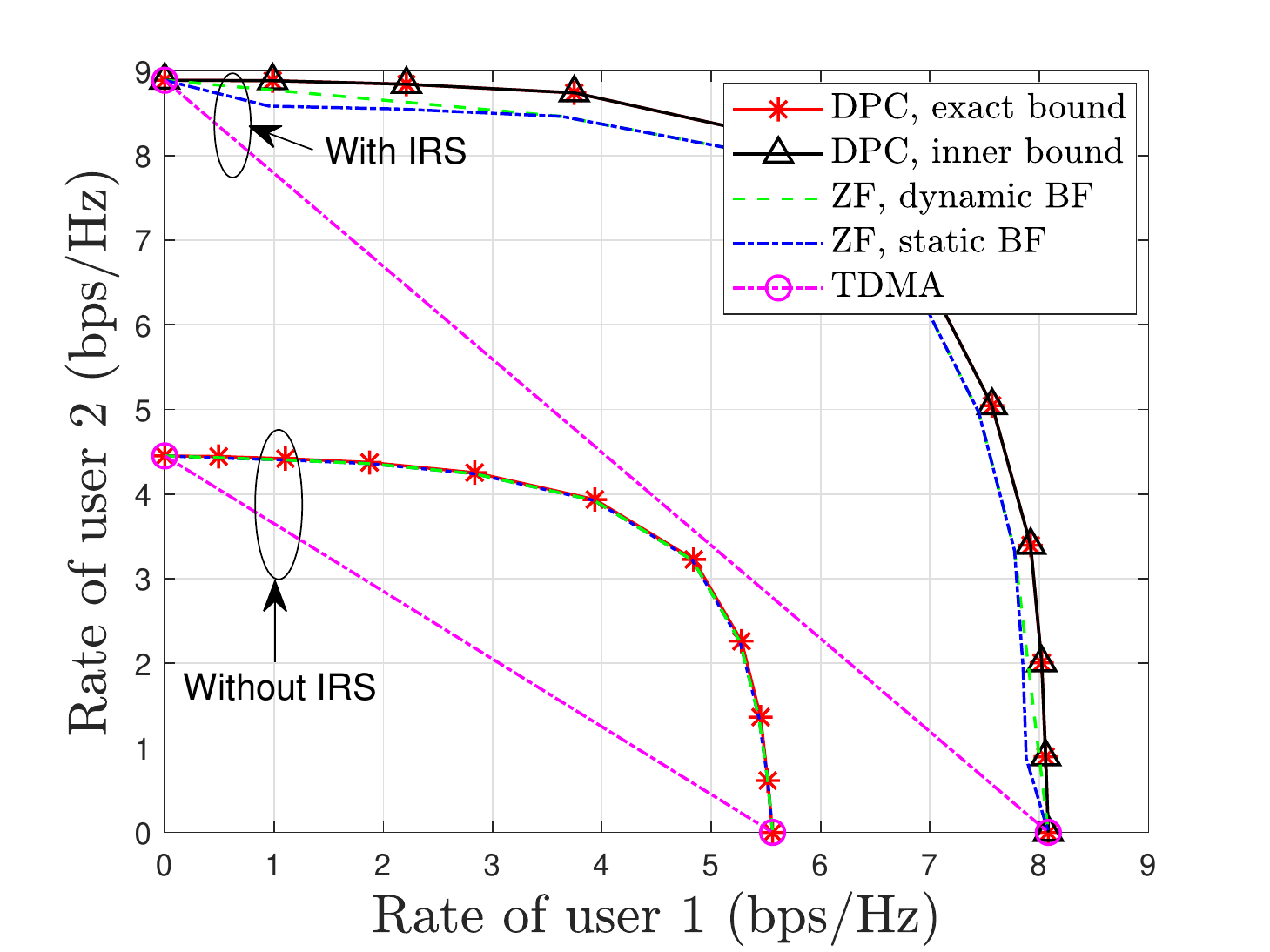}}
\subfigure[Resulting channel correlation.]{\label{correlation3}
\includegraphics[width= 2.5in, height=2in]{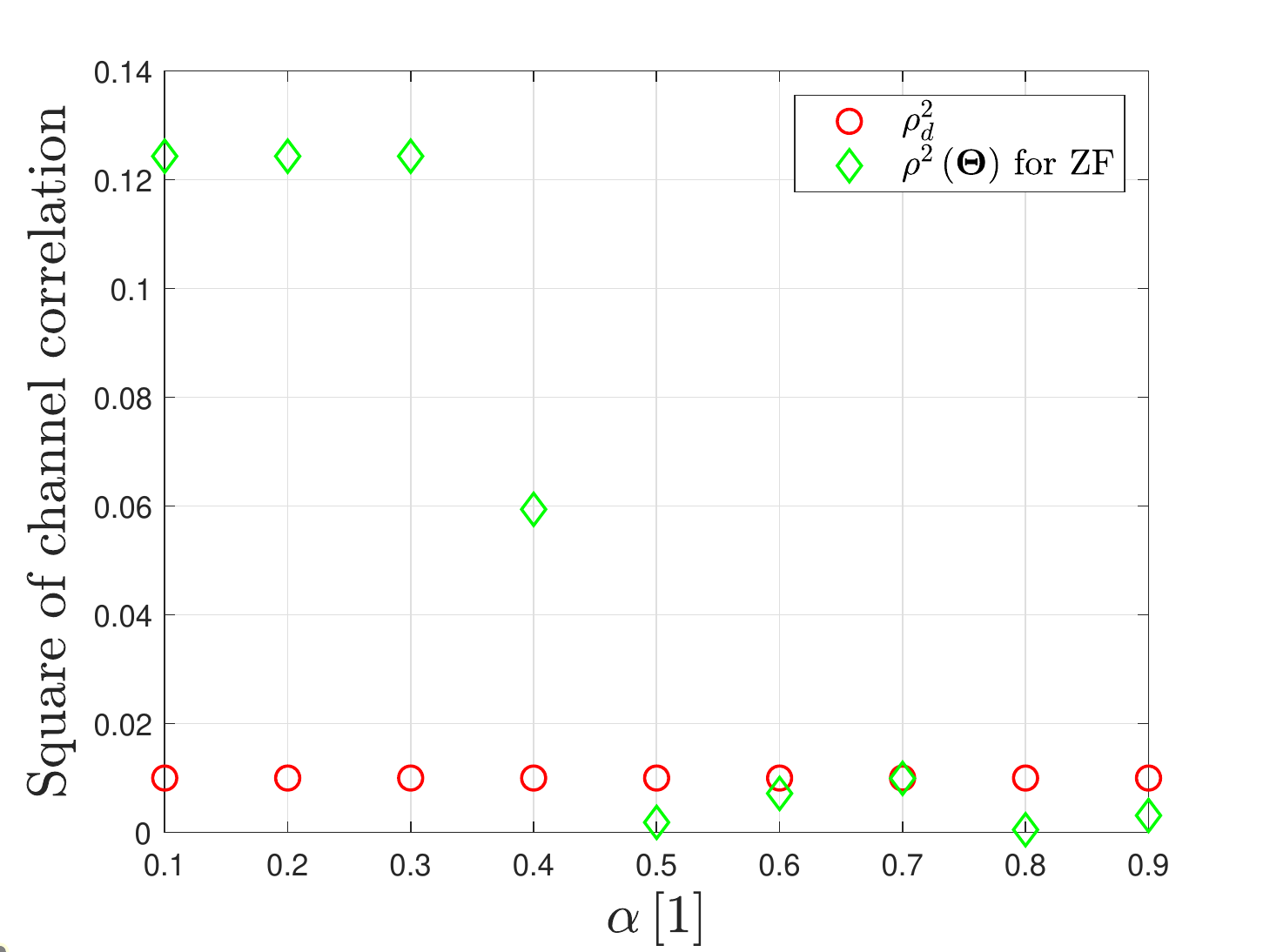}}
\setlength{\abovecaptionskip}{0.4cm}
\caption{{Capacity/rate region for IRS aided MISO broadcast channel under the case of $\rho _d^2 = 0.01$.}}\label{rate_region_moderate}
\vspace{-16pt}
\end{figure*}

\subsubsection{Moderate Channel Correlation} Then, we investigate the capacity/achievable rate regions under the moderate channel correlation of the direct links. Specifically, the capacity/achivable rate regions of the considered schemes under the case of $\rho _d^2 = 0.5$ are plotted in Fig. \ref{region22}. To capture the role of the IRS in such a scenario, we present the resulting channel correlation after optimizing IRS phase-shifts for ZF in Fig. \ref{correlation2}. For the system without IRS, it can be observed that employing dynamic beamforming is able to significantly enlarge the achievable rate region of ZF and the achievable rate region of ZF always contains that of the TDMA scheme. In addition, the rate gap between the ZF and DPC is still noticeable in this case. In contrast, the corresponding channel correlation can be effectively reduced by optimizing the IRS phase-shifts as demonstrated in Fig. \ref{correlation2}. This leads to the marginal performance gain achieved by dynamic beamforming over static beamforming. Moreover, the low channel correlation introduced by the IRS provides more spatial DoFs for reaping the high active beamforming gains for each user, which results in negligible performance-loss of ZF compared to DPC.

\subsubsection{Low Channel Correlation} Finally, we study the typical channel setup of low channel correlation by plotting the capacity/achivable rate regions under $\rho _d^2 = 0.01$. The corresponding ${\rho ^2}\left( {\bf{\Theta }} \right)$ after optimizing ${\bf{\Theta }}$ for ZF are also depicted in Fig. \ref{correlation3}. It is observed that both the capacity and rate regions achieved by DPC and ZF are nearly convex without IRS since the channels are quasi-orthogonal under the case of $\rho _d^2 = 0.01$. Thus, employing dynamic beamforming cannot bring the rate improvement and ZF is able to achieve almost the same performance as DPC, which validates our analysis in Proposition 4. In this case, the main bottleneck of the system performance is the channel strength rather than the channel correlation. To this end, by deploying IRS in such a system, the main objective of the IRS is to enhance the channel power gain regardless of reducing the channel correlation. As shown in Fig. \ref{correlation3}, ${\rho ^2}\left( {\bf{\Theta }} \right)$ is even higher than $\rho _d^2$ for some given values of ${\bm{\alpha }}[1]$, which results in slight performance loss of ZF compared to DPC and marginal performance gain attained by dynamic beamforming after integrating IRS.

For all the considered three channel setups with an IRS, the gain attained by adopting dynamic beamforming is negligible, which unveils that IRS is able to reduce the demand for implementing dynamic beamforming and thus the additional signalling overhead can be avoided.

\subsection{Sum Rate Performance}
In this subsection, we set $K = 4$ and ${\bm{\alpha }} = {\left[ {0.25,{\rm{0}}{\rm{.25,0}}{\rm{.25,0}}{\rm{.25}}} \right]^T}$. Under such a setup, the sum-rate performances of ZF and DPC for both cases with and without IRS are compared. All results are obtained over 100 independent channel realizations.

\begin{figure}
\begin{minipage}[t]{0.45\linewidth}
\centering
\includegraphics[width=2.9in]{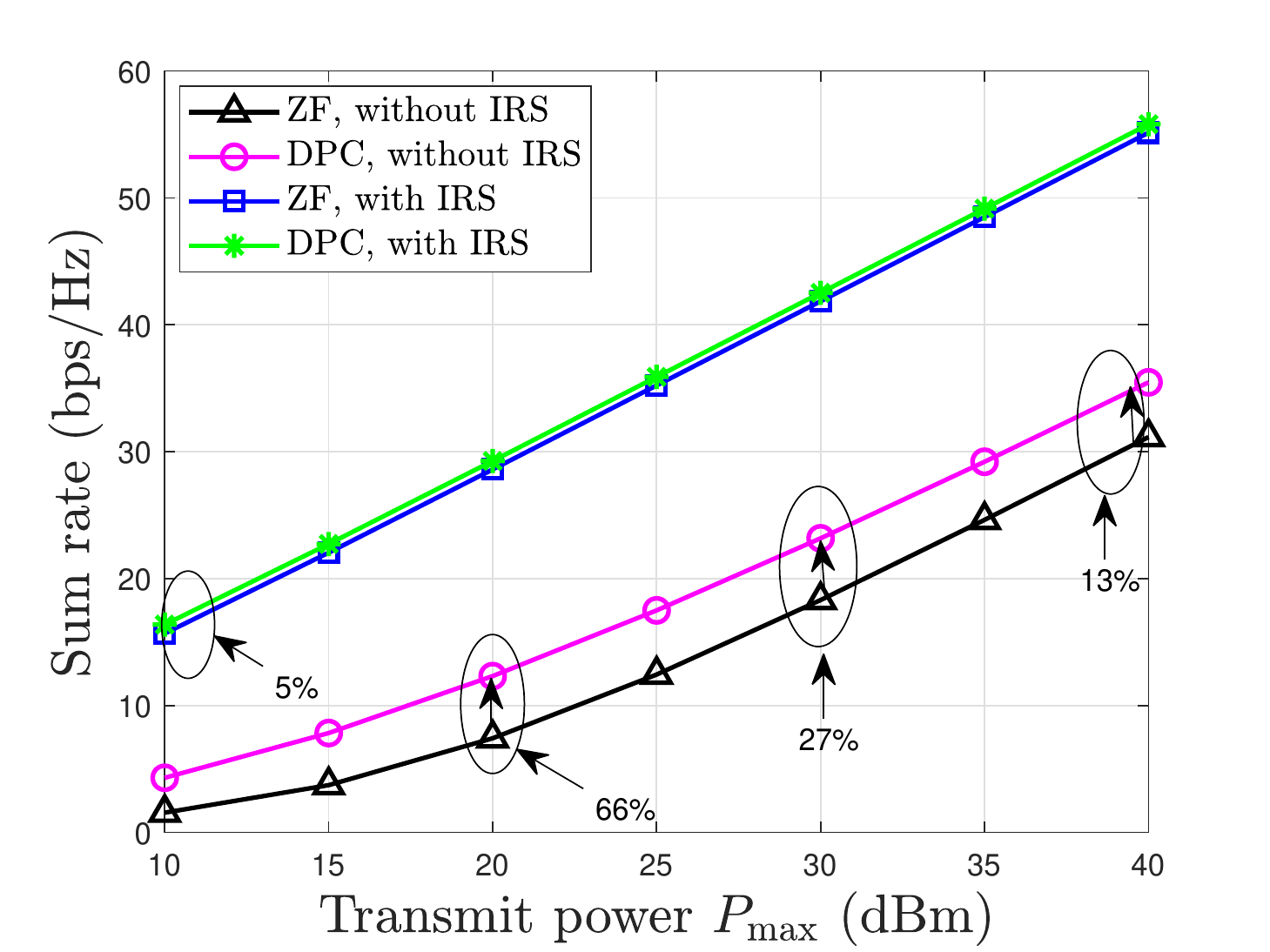}
\caption{The sum rate versus ${P_{\max }}$ for $N = 48$.}
\label{rate_power}
\end{minipage}%
\hfill
\begin{minipage}[t]{0.45\linewidth}
\centering
\includegraphics[width=2.9in]{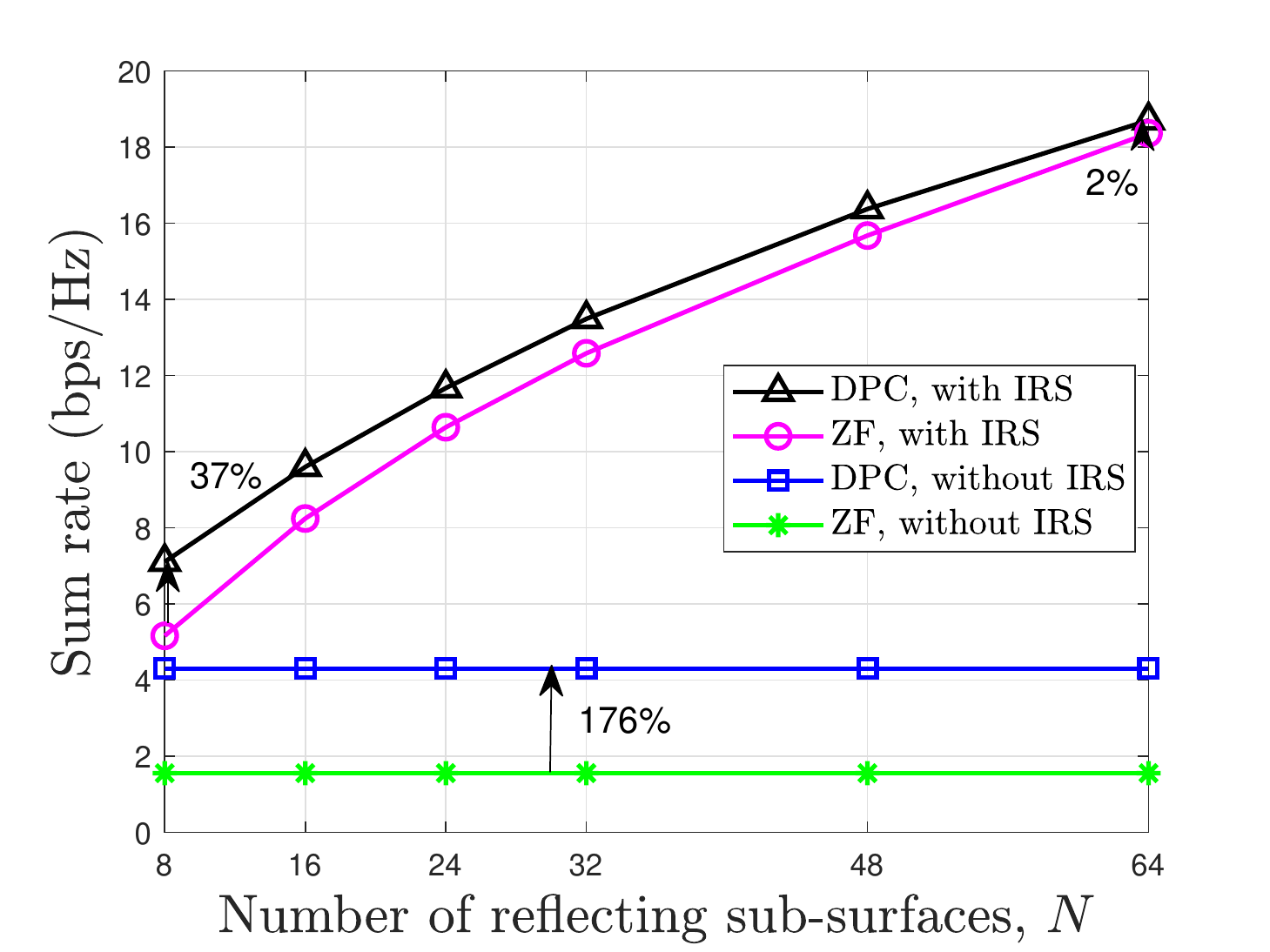}
\caption{The sum rate versus $N$ for ${P_{\max }} = 10$ dBm.}
\label{rate_N}
\end{minipage}
\vspace{-16pt}
\end{figure}

\subsubsection{Sum rate versus Transmit Power}
In Fig. \ref{rate_power}, we plot the system sum rate versus the maximum transmit power at the BS. In the scenario without IRS, it is observed that the sum rate of DPC is significantly higher than that of ZF. In particular, it is noticed that the relative performance gain of DPC over ZF, i.e., $\left( {{R_{{\rm{DPC}}}} - {R_{{\rm{ZF}}}}} \right)/{R_{{\rm{ZF}}}}$, decreases as ${P_{\max }}$ becomes large. However, even under the case of ${P_{\max }} = 40$ dBm, DPC can still achieve $13\% $ sum rate gains over the ZF scheme and thus the performance loss of ZF compared to DPC is considerable. By deploying the IRS, substantial sum rate improvement can be attained by the high passive beamforming gain and low-correlation channels provided by the IRS. To achieve an identical sum rate, a 15 dB performance gain can be achieved by the IRS aided ZF scheme over the DPC scheme without IRS. Moreover, the performance loss of the IRS aided ZF scheme compared to the IRS aided DPC scheme is negligible. Even under the small value of the transmit power, i.e., ${P_{\max }} = 10$ dBm, the corresponding performance loss is less than $5\% $.

\subsubsection{Sum rate versus $N$}
We study the impact of the number of reflecting sub-surfaces on the IRS aided MISO broadcast channel, by plotting the sum rate of all schemes versus $N$ in Fig. \ref{rate_N}. For all IRS aided schemes, the sum rate increases as $N$ increases while the sum rate of the schemes without IRS remains unchanged. In particular, the DPC scheme is able to achieve $176\%$ sum rate gains over the ZF scheme without IRS. After integrating the IRS in the system, the smart adjustment of IRS phase-shifts introduces a desirable ``double gain'' as it not only improves the channel strength but also reducing the channel correlation, thus rendering marginal performance loss of ZF compared to DPC. Specifically, it is observed from Fig. \ref{rate_N} that the relative sum rate gain of IRS aided DPC over IRS aided ZF is reduced from $37\%$ to $2\%$ by increasing $N$ from $8$ to $64$, which is consistent with our analysis in Theorem 1.

All results in Fig. \ref{rate_power} and Fig. \ref{rate_N} imply that the sum rate performance of IRS aided ZF is able to approach that of IRS aided DPC, which represents the performance limit. Compared to DPC, ZF is more appealing for practical systems due to its lower implementation complexity. Benefited by the ``smart radio environment'' introduced by the IRS, IRS aided ZF is expected to achieve near-optimal performance, which provides useful guidelines for simplifying the transceiver design.

\vspace{-10pt}
\section{Conclusion}
In this paper, we investigated the fundamental capacity limits of the IRS aided broadcast channel. Both the Pareto boundaries of the capacity and achievable rate regions for the IRS aided DPC and ZF schemes were characterized under the dynamic beamforming configurations. Then, we further studied two fundamental issues regarding the potential gains attained by dynamic beamforming and and the practical performance gap between the IRS aided ZF and DPC.  We rigorously proved that dynamic beamforming is able to enlarge the achievable rate region of ZF if designing IRS phase-shifts cannot achieve fully orthogonal channels, whereas the attained gains would become marginal as the reduction of channel correlations induced by smartly adjusting the IRS phase-shifts, which unveils that deploying IRS is able to reduce the demand for implementing dynamic beamforming and thus it is helpful for reducing signalling overhead. It was further analytically demonstrated the sum-rate achieved by the IRS aided ZF is capable of approaching that of the IRS aided DPC when the number of IRS elements becomes sufficiently large. The result is promising by indicating that the sophisticated transmission scheme can be replaced by the easy-implementation scheme at the cost of slight performance-loss. Extensive numerical results validated our theoretical findings and drew useful insights regarding the impact of the IRS on the transceiver design.

\section*{Appendix A: \textsc{Proof of Proposition 2}}
To prove Proposition 2, we first introduce the power minimization problem with respect to $\left\{ {{{\bf{S}}_k}} \right\}$ under any given ${\bf{\Theta }}$ as follows
\begin{subequations}\label{C3}
\begin{align}
\label{C3-a}\mathop {\min }\limits_{\left\{ {{{\bf{S}}_k}} \right\}}  \;\;&\sum\nolimits_{k = 1}^K {{\mathop{\rm Tr}\nolimits} } \left( {{{\bf{S}}_k}} \right)\\
\label{C3-b}{\rm{s.t.}}\;\;\;&{\bf{h}}_k^H{{\bf{S}}_k}{{\bf{h}}_k} \ge {\gamma _k}\left( {\frac{{{\bf{h}}_k^H\left( {\sum\nolimits_{i > k} {{{\bf{S}}_i}} } \right){{\bf{h}}_k}}}{{{\sigma ^2}}} + 1} \right), ~\forall {k},\\
\label{C3-e}&\eqref{C1-d}.
\end{align}
\end{subequations}
The optimal value of \eqref{C3} is denoted by ${p^*}$. Since problem \eqref{C3} is a convex optimization problem and satisfies Slater's condition, strong duality holds. The Lagrangian of \eqref{C3} is given by
\begin{align}\label{Lagrangian}
&{\cal L}\left( {{{\bf{S}}_k},{\lambda _k},{{\bf{Y}}_k}} \right) = \sum\nolimits_{k = 1}^K {{\rm{Tr}}} \left( {{{\bf{S}}_k}} \right) - \sum\nolimits_{k = 1}^K {{\rm{Tr}}\left( {{{\bf{Y}}_k}{{\bf{S}}_k}} \right)}\nonumber\\
&  + \sum\limits_{k = 1}^K {{\lambda _k}\left( {{\gamma _k}\left( {\frac{{{\bf{h}}_k^H\sum\nolimits_{i > k} {{{\bf{S}}_i}{{\bf{h}}_k}} }}{{{\sigma ^2}}} + 1} \right) - \frac{{{\bf{h}}_k^H{{\bf{S}}_k}{{\bf{h}}_k}}}{{{\sigma ^2}}}} \right)}\nonumber\\
& = \sum\nolimits_{k = 1}^K {\left( {{\rm{Tr}}\left( {\left( {{{\bf{Z}}_k} - {{\bf{Y}}_k}} \right){{\bf{S}}_k}} \right) + {\lambda _k}{\gamma _k}} \right)},
\end{align}
where ${\lambda _k} \ge 0$ and ${{\bf{Y}}_k} \succeq {\bf{0}}$ are Lagrange multipliers associated with constraints \eqref{C3-b} and \eqref{C1-d}, and
\begin{align}\label{Z_k}
{{\bf{Z}}_k} = {{\bf{I}}_M}  - {\lambda _k}\frac{{{{\bf{h}}_k}{\bf{h}}_k^H}}{{{\sigma ^2}}} + \sum\nolimits_{i = 1}^{k - 1} {\frac{{{\lambda _i}{\gamma _i}{{\bf{h}}_i}{\bf{h}}_i^H}}{{{\sigma ^2}}}}.
\end{align}
Then, the Lagrange dual function of \eqref{C3} is given by
\begin{align}\label{dual}
&g\left( {{\lambda _k},{{\bf{Y}}_k}} \right) = \mathop {\inf }\limits_{{{\bf{S}}_k}} {\cal L}\left( {{{\bf{S}}_k},{\lambda _k},{{\bf{Y}}_k}} \right) = \begin{cases}{\sum\nolimits_{k = 1}^K {{\lambda _k}{\gamma _k}}}, &{{\rm{if }}~{{\bf{Z}}_k} - {{\bf{Y}}_k} = {\bf{0}},\forall k}, \cr { - \infty },
&{{\rm{otherwise}}}. \end{cases}
\end{align}
Hence, the dual problem of \eqref{C3} can be written as
\begin{subequations}\label{C4}
\begin{align}
\label{C4-a}\mathop {\max }\limits_{\left\{ {{\lambda _K}} \right\}, \left\{ {{{\bf{Y}}_k}} \right\}}  \;\;&{\sum\nolimits_{k = 1}^K {{\lambda _k}{\gamma _k}} }\\
\label{C4-b}{\rm{s.t.}}\;\;\;\;\;\;&{{\bf{Z}}_k} - {{\bf{Y}}_k} = {\bf{0}},\forall k,\\
\label{C4-c}&{\lambda _k} \ge 0,{{\bf{Y}}_k} \succeq {\bf{0}},\forall k.
\end{align}
\end{subequations}
The optimal solution of problem \eqref{C4} is denoted by $\left\{ {\lambda _k^*} \right\}$ and we have ${p^*} = \sum\nolimits_{k = 1}^K {\lambda _k^*}{\gamma _k}$ since strong duality of problem \eqref{C3} holds. Then, we focus on obtaining $\left\{ {\lambda _k^*}\right\}$. Applying Karush-Kuhn-Tucker conditions of \eqref{C3} yields
\begin{align}\label{KKT}
&{{\bf{I}}_M} - {{\bf{Y}}_k} - \lambda _k^*\frac{{{{\bf{h}}_k}{\bf{h}}_k^H}}{{{\sigma ^2}}} + \sum\nolimits_{i = 1}^{k - 1} {\frac{{\lambda _i^*{\gamma _i}{{\bf{h}}_i}{\bf{h}}_i^H}}{{{\sigma ^2}}}}  = {\bf{0}},\forall k,\\
&{\rm{Tr}}\left( {{{\bf{Y}}_k}{\bf{S}}_k^*} \right) = 0,\forall k.
\end{align}
Multiplying both sides of \eqref{KKT} by ${{\bf{S}}_k^*}$, we have
\begin{align}\label{optimal_SK}
{\bf{S}}_k^* - \lambda _k^*\frac{{{{\bf{h}}_k}{\bf{h}}_k^H}}{{{\sigma ^2}}}{\bf{S}}_k^* + \sum\nolimits_{i = 1}^{k - 1} {\frac{{\lambda _i^*{\gamma _i}{{\bf{h}}_i}{\bf{h}}_i^H}}{{{\sigma ^2}}}} {\bf{S}}_k^* = {\bf{0}},\forall k.
\end{align}
It was shown in \cite{5233822} that there always exists an optimal solution for all ${\bf{S}}_k^*$'s of rank one regarding problem \eqref{C3}. Then, we can always find ${\bf{v}}_k^* \in {\mathbb{C}^{M \times 1}}$ and ${\bf{S}}_k^* = {\bf{v}}_k^*{\left( {{\bf{v}}_k^*} \right)^H}$. Then, \eqref{optimal_SK} can be re-expressed as
\begin{align}\label{optimal_vK}
{\bf{v}}_k^* - \lambda _k^*\frac{{{{\bf{h}}_k}{\bf{h}}_k^H}}{{{\sigma ^2}}}{\bf{v}}_k^* + \sum\nolimits_{i = 1}^{k - 1} {\frac{{\lambda _i^*{\gamma _i}{{\bf{h}}_i}{\bf{h}}_i^H}}{{{\sigma ^2}}}} {\bf{v}}_k^* = {\bf{0}},\forall k,
\end{align}
since ${\bf{v}}_k^* \ne {\bf{0}}$. From \eqref{optimal_vK}, we obtain
\begin{align}\label{optimal_vk1}
{\bf{v}}_k^* = {\left( {{{\bf{I}}_M} + \sum\nolimits_{i = 1}^{k - 1} {\frac{{\lambda _i^*{\gamma _i}{{\bf{h}}_i}{\bf{h}}_i^H}}{{{\sigma ^2}}}} } \right)^{ - 1}}\lambda _k^*\frac{{{{\bf{h}}_k}{\bf{h}}_k^H}}{{{\sigma ^2}}}{\bf{v}}_k^*,\forall k.
\end{align}
Multiplying both sides of \eqref{optimal_vk1} by ${{\bf{h}}_k^H}$, we have
\begin{align}\label{lamda_expression}
{\bf{h}}_k^H{\left( {{{\bf{I}}_M} + \sum\nolimits_{i = 1}^{k - 1} {\frac{{\lambda _i^*{\gamma _i}{{\bf{h}}_i}{\bf{h}}_i^H}}{{{\sigma ^2}}}} } \right)^{ - 1}}{{\bf{h}}_k}\frac{{\lambda _k^*}}{{{\sigma ^2}}} = 1,\forall k.
\end{align}
Based on \eqref{lamda_expression}, the optimal value of \eqref{C3} is derived in \eqref{power_value}-\eqref{power_value2}, which completes the proof.

\section*{Appendix B: \textsc{Proof of Lemma 2}}
Let ${\bf{A}} = {\bf{\Theta G}}{{\bf{G}}^H}{{\bf{\Theta }}^H}$. Under the condition that ${\bf{G}} \sim {\cal C}{\cal N}\left( {0,{\bf{I}}} \right)$ and ${\bf{\Theta }}{\rm{ = }}{\mathop{\rm diag}\nolimits} \left( {{e^{j{\theta _1}}}, \ldots ,{e^{j{\theta _N}}}} \right)$, it can be readily verified that ${\mathop{\rm rank}\nolimits} \left( {\bf{A}} \right) = M$. Then, we derive ${{\mathop{\mathbb{E}}\nolimits} \left\{ {{{\left| {{\bf{h}}_{r,k}^H{\bf{\Theta G}}{{\bf{G}}^H}{{\bf{\Theta }}^H}{{\bf{h}}_{r,m}}} \right|}^2}} \right\}}$ as follows
\begin{align}\label{expectation1}
&{\mathop{\mathbb{E}}\nolimits} \left\{ {{{\left| {{\bf{h}}_{r,k}^H{\bf{\Theta G}}{{\bf{G}}^H}{{\bf{\Theta }}^H}{{\bf{h}}_{r,m}}} \right|}^2}} \right\} \nonumber\\
&= {\mathop{\mathbb{E}}\nolimits} \left\{ {{\rm{Tr}}\left( {{\bf{A}}{{\bf{h}}_{r,m}}{\bf{h}}_{r,m}^H{{\bf{A}}^H}{{\bf{h}}_{r,k}}{\bf{h}}_{r,k}^H} \right)} \right\}\nonumber\\
& = {\mathop{\rm Tr}\nolimits} \left( {{\bf{A}}{\mathop{\mathbb{E}}\nolimits} \left\{ {{{\bf{h}}_{r,m}}{\bf{h}}_{r,m}^H} \right\}{{\bf{A}}^H}{\mathop{\mathbb{E}}\nolimits} \left\{ {{{\bf{h}}_{r,k}}{\bf{h}}_{r,k}^H} \right\}} \right) = {\mathop{\rm Tr}\nolimits} \left( {{\bf{A}}{{\bf{A}}^H}} \right) = \sum\limits_{m = 1}^M {\lambda _{A,m}^2},
\end{align}
where ${\lambda _{A,m}}$'s are $M$ positive eigenvalues of ${\bf{A}}$. Similarly, ${{\mathop{\mathbb{E}}\nolimits} \left\{ {{{\left\| {{\bf{h}}_{r,k}^H{\bf{\Theta G}}} \right\|}^2}} \right\}}$ and ${{\mathop{\mathbb{E}}\nolimits} \left\{ {{{\left\| {{\bf{h}}_{r,m}^H{\bf{\Theta G}}} \right\|}^2}} \right\}}$ can be calculated as
\begin{align}\label{expectation2}
{\mathop{\mathbb{E}}\nolimits} \left\{ {{{\left\| {{\bf{h}}_{r,k}^H{\bf{\Theta G}}} \right\|}^2}} \right\} = {\mathop{\mathbb{E}}\nolimits} \left\{ {{{\left\| {{\bf{h}}_{r,m}^H{\bf{\Theta G}}} \right\|}^2}} \right\} = \sum\nolimits_{m = 1}^M {{\lambda _{A,m}}}.
\end{align}
Hence, we can obtain
\begin{align}\label{correlation_lower}
\frac{{{\mathop{\mathbb{E}}\nolimits} \left\{ {{{\left| {{\bf{h}}_{r,k}^H{\bf{\Theta G}}{{\bf{G}}^H}{{\bf{\Theta }}^H}{{\bf{h}}_{r,m}}} \right|}^2}} \right\}}}{{{\mathop{\mathbb{E}}\nolimits} \left\{ {{{\left\| {{\bf{h}}_{r,k}^H{\bf{\Theta G}}} \right\|}^2}} \right\}{\mathop{\mathbb{E}}\nolimits} \left\{ {{{\left\| {{\bf{h}}_{r,m}^H{\bf{\Theta G}}} \right\|}^2}} \right\}}} = \frac{{\sum\nolimits_{m = 1}^M {\lambda _{A,m}^2} }}{{{{\left( {\sum\nolimits_{m = 1}^M {{\lambda _{A,m}}} } \right)}^2}}} \mathop  \ge \limits^{\left( a \right)}  \frac{1}{M},
\end{align}
where (a) is due to the Cauchy-Schwarz inequality. Lemma 2 is thus proved.


\bibliographystyle{IEEEtran}

\bibliography{IEEEabrv,myref}


\end{document}